\documentclass[default]{sn-jnl}

\usepackage{graphicx}
\usepackage{multirow}
\usepackage{amsmath,amssymb,amsfonts}
\usepackage{amsthm}
\usepackage{mathrsfs}
\usepackage[title]{appendix}
\usepackage{xcolor}
\usepackage{textcomp}
\usepackage{manyfoot}
\usepackage{booktabs}
\usepackage{algorithm}
\usepackage{algorithmicx}
\usepackage{algpseudocode}
\usepackage{listings}
\usepackage{mathtools}
\usepackage{ulem}
\usepackage{bm}
\usepackage{makecell}
\usepackage{cases}

\theoremstyle{thmstyleone}%
\newtheorem{lemma}{Lemma}
\newtheorem{theorem}{Theorem}
\theoremstyle{thmstyletwo}%
\newtheorem{remark}{Remark}%

\theoremstyle{thmstylethree}%

\DeclareMathOperator*{\Wt}{wt}

\begin{document}
\title[New constructions of MDS symbol-pair codes via simple-root cyclic codes]{New constructions of MDS symbol-pair codes via simple-root cyclic codes\footnote{Corresponding author: Weijun Fang}}

\author[1,2,3]{\fnm{Rongxing} \sur{Qiu}}\email{qrx@mail.sdu.edu.cn}

\author[1,2,3]{\fnm{Weijun} \sur{Fang}}\email{fwj@sdu.edu.cn}

\affil[1]{\orgname{State Key Laboratory of Cryptography and Digital Economy Security, Shandong University}, \orgaddress{\city{Qingdao}, \postcode{266237}, \country{China}}}

\affil[2]{\orgdiv{Key Laboratory of Cryptologic Technology and Information Security}, \orgname{Ministry of Education, Shandong University}, \orgaddress{\city{Qingdao}, \postcode{266237}, \country{China}}}


\affil[3]{\orgdiv{School of Cyber Science and Technology}, \orgname{Shandong University}, \orgaddress{\city{Qingdao}, \postcode{266237}, \country{China}}}



\abstract{In modern storage technologies, symbol-pair codes have emerged as a crucial framework for addressing errors in channels where symbols are read in overlapping pairs to guard against pair errors. A symbol-pair code that meets the Singleton-type bound is called a maximum distance separable (MDS) symbol-pair code. MDS symbol-pair codes are optimal in the sense that they have the highest pair error-correcting capability. In this paper, we focus on new constructions of MDS symbol-pair codes using simple-root cyclic codes. Specifically,  three new infinite families of $(n, d_P)_q$-MDS symbol-pair codes are obtained: (1) $(n=4q+4,d_P=7)_q$ for $q\equiv 1\pmod 4$; (2) $(n=4q-4,d_P=8)_q$ for $q\equiv 3\pmod 4$; (3) $(n=2q+2,d_P=9)_q$ for $q$ being an odd prime power. The first two constructions are based on analyzing the solutions of certain equations over finite fields. The third construction arises from the decomposition of cyclic codes, where we utilize the orthogonal relationships between component codes and their duals to rigorously exclude the presence of specific codewords. It is worth noting that for the pair distance $d_P=7$ or $8$, our $q$-ary MDS symbol-pair codes achieve the longest known code length when $q$ is not a prime. Furthermore, for $d_P=9$, our codes attain the longest code length regardless of whether $q$ is prime or not.}



\maketitle

\section{Introduction}
\label{sec:intro}
    In high-density data storage systems, symbols in certain channels cannot always be performed individually and are read in possibly overlapping pairs. Considering the lower read resolution of magnetic-soft channels, the symbol-pair channel model has been proposed as an alternative to the individual-symbol channel. To solve this problem, Cassuto and Blaum in \cite{Cassuto2010,Cassuto2011}, followed by Cassuto and Litsyn in \cite{Cassuto2011a}, proposed and gave early constructions of symbol-pair codes. In this new framework, the outputs and the errors are no longer individual symbols, but rather overlapping pairs of adjacent symbols. The fundamental distinction of symbol-pair codes lies in their unique distance metric, i.e. the symbol-pair distance. Unlike conventional Hamming distance that counts single-symbol discrepancies, this metric measures the minimum number of distinct symbol pairs between any two codewords. This paradigm shift necessitates novel code constructions that optimize pairwise correlations while maintaining efficient encoding/decoding complexity. 
    
    Let $n$ be a positive integer and $\mathbb{F}_q$ be a finite field with $q$ elements, where $q$ be a prime power.  For any fixed vector $\mathbf{x} = (x_0,x_1,\cdots,x_{n-1}) \in \mathbb{F}_q^n$,  the symbol-read vector of $\mathbf{x}$ is
    $$\pi(\mathbf{x}) = \big((x_0,x_1),(x_1,x_2),\cdots,(x_{n-2},x_{n-1}),(x_{n-1},x_0)\big).$$ Obviously, each vector $\mathbf{x}$ has a unique symbol-pair 
    representation $\pi(\mathbf{x})$. Recall that the Hamming weight of $\mathbf{x}$ is defined as
    $$w_H(\mathbf{x}) = | \left\{0 \leq    i \leq    n-1: x_i \neq  0\right\}|.$$ 
    Correspondingly, the symbol-pair weight of $\mathbf{x}$ is defined as
    $$w_P(\mathbf{x}) = | \left\{0 \leq    i \leq    n-1: (x_i,x_{i+1}) \neq  (0,0)\right\}|,$$
    where the subscripts are taken modulo $n$. For any two vectors $\mathbf{x} = (x_0,x_1,\cdots,x_{n-1})$ and $\mathbf{y} = (y_0,y_1,\cdots,y_{n-1})$  in $\mathbb{F}_q^n$, the Hamming distance between $\mathbf{x}$ and $\mathbf{y}$ is defined as 
    $$d_H(\mathbf{x},\mathbf{y}) = |\left\{0\leq i \leq n-1:x_i\neq y_i\right\},$$ and the pair distance between $\mathbf{x}$ and $\mathbf{y}$ is defined as 
    $$d_P(\mathbf{x},\mathbf{y}) = d_H\big(\pi(\mathbf{x}),\pi(\mathbf{y})\big) =|\left\{0\leq i \leq n-1:(x_i,x_{i+1})\neq (y_i,y_{i+1})\right\}|,$$ where the subscripts are taken modulo $n$. A code $\mathcal{C}$ over $\mathbb{F}_q$ of length $n$ is a subspace of $\mathbb{F}_q^n$. We denote a symbol-pair code $\mathcal{C}$ by $(n,M,d_P)$ when $\mathcal{C}$ has size $M$ and minimum pair distance 
    $$d_P(\mathcal{C}) = \min \left\{d_P(\mathbf{a},\mathbf{b}) : \mathbf{a},\mathbf{b} \in \mathcal{C}, \mathbf{a}\neq \mathbf{b}\right\}.$$ 
    The error-correcting capability of a symbol-pair code closely resembles that of a classical code as it is measured in terms of the minimum pair distance. In \cite{Cassuto2011}, Cassuto and Blaum showed that an $(n,M,d_P)_q$-code $\mathcal{C}$ can correct up to $t$ pair errors if and only if $d_P \geq 2t+1$. Therefore, for fixed $n$ and $M$, one wishes to construct symbol-pair codes with minimum pair distance as large as possible. In 2012, Chee $et$ $al.$ in \cite{Chee2012} derived a Singleton-type bound for $(n,M,d_P)_q$-code: 
    \begin{equation}
        M\leq q^{n-d_P+2}.
        \label{singleton_type_bound}
    \end{equation}
    If $M=q^{n-d_P+2}$, the $(n,M,d_P)$-code is called a maximum distance separable (MDS) symbol-pair code. Since $M$ is completely determined by $q$, $n$ and $d_P$, an MDS $(n, M, d_P )_q$ symbol-pair code is denoted by $(n, d_P )_q$ for convenience.
    
    MDS symbol-pair codes have excellent symbol-pair errorcorrecting capability and are the most useful and interesting symbol-pair codes. Thus, explicit constructions of MDS symbol-pair codes are both theoretically and practically significant. Many attempts have been made in the construction of MDS symbol-pair codes. In \cite{Chee2013},  Chee $et$ $al.$ obtained many classes of MDS symbol-pair codes by using classical MDS codes and interleaving and extending classical MDS codes. Moreover, they investigated MDS symbol-pair codes with length $(q^2 + 2q)/2$ using Eulerian graphs. In 2018, Ding $et$ $al.$ \cite{Ding2018} constructed some MDS symbol-pair codes based on projective geometry and elliptic curves. In addition, they showed that there exist $q$-ary MDS symbol-pair codes from algebraic geometry codes over elliptic curves with larger minimum symbol-pair distance and the lengths of these codes are bounded by $q + 2\sqrt{q}$.

    Cyclic and constacyclic codes are powerful ingredients in the construction of MDS symbol-pair codes because of their rich algebraic structures. By using repeated root cyclic or constacyclic codes, many families of MDS symbol pair codes with minimum pair distances $d_P \leq 12$ were constructed in \cite{Chen2017,Dinh2021,Dinh2020,Dinh2019,Dinh2018,Kai2018,Li2023,Sun2017,Zhao2021,Ma2022,Ma2022a}. The benefit of repeated-root codes in the construction is from their concatenated structure. However, the sizes of the fields defined the resulting codes were generally required to be prime numbers. Using simple-root constacyclic codes, Kai $et$ $al.$ in \cite{Kai2015} proposed MDS symbol-pair codes of minimum pair distance $d_P=5$ or $6$ from almost-MDS constacyclic codes. Following their idea, Li and Ge in \cite{Li2017} presented three new classes of MDS symbol-pair codes, also with minimum pair distances $d_P \in \{5, 6\}$. Moreover, they also constructed a number of MDS symbol-pair codes with minimum symbol-pair distance $d_P=7$ by analyzing certain linear fractional transformations. The advantage of simple-root codes is that the construction can be optimized in both the pair distance and the Hamming distance. However, it is generally difficult to construct MDS symbol-pair codes over $\mathbb{F}_q$ with length more than $q + 1$ and minimum pair distance at least $7$. In \cite{Luo2023}, Luo et al. used matrix-product codes to construct four families of MDS symbol-pair codes including $(2q+2,7)_q$ MDS symbol-pair codes for even prime power $q$. Very recently, Kai $et$ $al.$ in \cite{Kai2024} constructed two classes of MDS symbol-pair codes with minimum pair distance $d_P=7$ through the decomposition of cyclic codes and analyzing certain equations over finite fields. In Table \ref{Tab:known_result}, we provide known constructions of MDS symbol-pair codes.
\begin{table}[ht]
\centering
\caption{Known constructions of $(n,d_P)_q$-MDS symbol-pair codes, where $q=p^t$}
\begin{tabular}{|c|c|c|c|l|}
\hline
No. & $d_P$ & $n$ & $q = p^t$ & Reference \\
\hline
1  & $d_P = 5$  & $5 \leq n \leq q^2 + q + 1$ & -- & \cite{Kai2015,Ding2018}\\
\hline
2  & \multirow{6}{*}{$d_P = 6$} & $\max\{6, q + 2\} \leq n \leq q^2$ & --  & \cite{Chen2017,Ding2018} \\
\cline{1-1}\cline{3-5}
3  &  & $n = q^2 + 1$ & -- & \cite{Kai2015} \\
\cline{1-1}\cline{3-5}
4  &  & $n = p^2 + p$ & $q = p > 2$ & \cite{Kai2018} \\
\cline{1-1}\cline{3-5}
5  &  & $n = 2p^2 - 2p$ & $q = p > 2$ & \cite{Kai2018} \\
\cline{1-1}\cline{3-5}
6  &  & $n = 3m$ and $m \in [3, q]$ & -- & \cite{Luo2023} \\
\cline{1-1}\cline{3-5}
7  &  & $n = q^2 + q$ & $q \neq 4, 5, 2^t$ for odd $t$ & \cite{Luo2023} \\
\hline
8  & \multirow{8}{*}{$d_P = 7$}  & $n = 3p$ & $q = p \geq 5$ & \cite{Chen2017} \\
\cline{1-1}\cline{3-5}
9  &  & $n = 4p$ & $q = p \geq 5$ & \cite{Kai2018,Ma2022} \\
\cline{1-1}\cline{3-5}
10 &  & $n = 5p$ & $q = p$ and $5 \mid (p-1)$ & \cite{Ma2022} \\
\cline{1-1}\cline{3-5}
11 &  & $n = 3mp$ with $m \in [1, q/p]$ & $3\mid(q-1)$ and $2 \nmid q$  & \cite{Luo2023} \\
\cline{1-1}\cline{3-5}
12 &  & $n = 2q + 2$ & $2\mid q$ & \cite{Luo2023} \\
\cline{1-1}\cline{3-5}
13 &  & $n = 2q + 2$ & $p > 2$ & \cite{Kai2024} \\
\cline{1-1}\cline{3-5}
14 &  & $n = 4q - 4$ & $4\mid(q+1)$ & \cite{Kai2024} \\
\cline{1-1}\cline{3-5}
15 &  & $n = 4q + 4$ & $4 \mid (q - 1)$ & Theorem \ref{Thm:theorem1} \\
\hline
16 & \multirow{4}{*}{$d_P = 8$}  & $n = 3p$ & $q = p$ and $3 \mid (p-1)$ & \cite{Chen2017} \\
\cline{1-1}\cline{3-5}
17 &  & $n = 5p$ & $q = p$ and $5 \mid (p-1)$ & \cite{Ma2022} \\
\cline{1-1}\cline{3-5}
18 &  & $n = 2m$ with $m \leq q + 2$ & $p = 2$ & \cite{Chee2013} \\
\cline{1-1}\cline{3-5}
19 &  & $n = 4q - 4$ & $4 \mid (q + 1)$ & Theorem \ref{Thm:theorem2} \\
\hline
20 & $d_P = 9$ & $n = 2q + 2$ & $p>2$ & Theorem \ref{Thm:theorem3} \\
\hline
21 & \multirow{2}{*}{$d_P = 10$} & $n = 3p$ & $q = p$ and $3 \mid (p-1)$ & \cite{Ma2022a} \\
\cline{1-1}\cline{3-5}
22&   &  $n=3mp$ with $m\in [1,q/p]$ & $3\mid(q-1)$ &\cite{Luo2023}\\
\hline
23 & $d_P = 12$ & $n = 3p$ & $q = p$ and $3 \mid (p-1)$ & \cite{Ma2022a} \\
\hline
24 & $d_P = \ell$  & $4 \leq \ell \leq n \leq q + 1$ & -- & \cite{Chee2013} \\
\hline
25 & $d_P = 2\ell$ & $n = 2m$ with $4 \leq \ell+1 \leq m \leq q + 1$ & -- & \cite{Chee2013} \\
\hline
26 & $d_P = 2\ell+1$  & $n = 2m$ with $\ell + 1 \leq m \leq q$ & -- & \cite{Luo2023} \\
\hline
27 & $d_P = 2q-1$  & $n = 2q + 2$ & $p=2$ & \cite{Luo2023} \\
\hline
28 & $d_P = n - 4$  & $n = 2m$ with $m \leq q + 2$ & $p = 2$ & \cite{Chee2013} \\
\hline
\end{tabular} \label{Tab:known_result}
\end{table}

    In this paper, we construct three new classes of $q$-ary $(n, d_P)$-MDS symbol-pair codes as follows:
    \newline (1) $(n=4q+4,d_P=7)_q$ for $q\equiv 1\pmod 4$;
    \newline (2) $(n=4q-4,d_P=8)_q$ for $q\equiv 3\pmod 4$;
    \newline (3) $(n=2q+2,d_P=9)_q$ for $q$ is an odd prime power.
    \newline 
    The first two constructions are obtained by analyzing the solutions of certain equations over finite fields. In contrast, the third construction arises from the decomposition of cyclic codes. Utilizing the orthogonal relationships between component codes and their duals, we rigorously exclude the presence of specific codewords in code $\mathcal{C}$ by resolving complex equation systems over finite fields. According to Table \ref{Tab:known_result},  for the pair minimum distance 
$d_P=7, 8$, when $q$ is not a prime, our 
$q$-ary MDS symbol-pair codes achieve the longest known code length. Moreover, for 
$d_P=9$, our construction provides the longest MDS symbol-pair codes currently known.

 The rest of this paper is organized as follows. Section \ref{sec:preliminaries} introduces some related notation and known results on constacyclic codes and the decomposition of cyclic codes. In Section \ref{sec3} we construct three classes of MDS symbol-pair code through the decomposition of cyclic codes and analyzing the solutions of certain equations over finite fields. Section \ref{sec4} concludes the paper.

\section{Preliminaries}
\label{sec:preliminaries}
In this section, we introduce some preliminaries which will be used in the subsequent sections.
\subsection{Constacyclic Codes}
\label{sec:constacyclic_codes}
    Let $q$ be an odd prime power and $\mathbb{F}_q$ be the finite field with $q$ elements. Let $\mathbb{F}_q^* = \mathbb{F}_q \backslash \{0\} $ and $\lambda \in \mathbb{F}_q^*$. A $\lambda$-constacyclic code $\mathcal{C}$ of length $n$ over $\mathbb{F}_q$ is a linear code with the property that $(\lambda c_{n-1}, c_0,c_1,\cdots,c_{n-2})\in \mathcal{C}$ whenever $(c_0,c_1,\cdots,c_{n-2},c_{n-1})\in \mathcal{C}$. The code $\mathcal{C}$ is called cyclic if $\lambda = 1$ and negacyclic if $\lambda=-1$.
    
    If we identify any vector $(c_0,c_1,\cdots,c_{n-1})\in \mathbb{F}_q^n$ with the polynomial
    $\sum_{i=0}^{n-1} c_ix^i \in \mathbb{F}_q[x] /\langle x^n-\lambda\rangle,$
    any code of length $n$ over $\mathbb{F}_q$ corresponds to a subset of the quotient ring $\mathbb{F}_q[x] /\langle x^n-\lambda\rangle$. It can be easily checked that a linear code is $\lambda$-constacyclic if and only if the corresponding subset in $\mathbb{F}_q[x] / \langle x^n-\lambda\rangle$ is an ideal of the ring $\mathbb{F}_q[x] / \langle x^n-\lambda\rangle$. It is well known that every ideal of $\mathbb{F}_q[x] / \langle x^n-\lambda\rangle$ is principal. So there exists a monic polynomial $g(x)$ such that $\mathcal{C}=\langle g(x) \rangle$ and $g(x)$ has the smallest degree among all the generators of $\mathcal{C}$. Then $g(x)$ is unique and called the \textit{generator polynomial} of $\mathcal{C}$. Let $h(x) = \frac{x^n-\lambda}{g(x)}$, which is known as the \textit{check polynomial} of $\mathcal{C}$. Moreover, $\dim(\mathcal{C}) = n - \deg(g(x))$.

    Let $f(x) = \sum_{i=0}^{k} a_ix^i \in \mathbb{F}_q[x]$ be a polynomial of degree $k$, where $a_k \neq 0$ and $a_i\in \mathbb{F}_q$ for $0\leq i \leq k-1$. Define the \textit{reciprocal polynomial} $f_R(x)$ by $$f_R(x) = x^kf(x^{-1}) = \sum_{i=0}^{k} a_{k-i}x^i. $$ Given two vectors $\mathbf{x} =(x_0,x_1,\cdots,x_{n-1}) \in \mathbb{F}_q$ and $\mathbf{y} =(y_0,y_1,\cdots,y_{n-1}) \in \mathbb{F}_q$, the Euclidean inner product on $\mathbb{F}_q$ is defined by $$\mathbf{x}\cdot\mathbf{y} = \sum_{i=1}^{n-1} x_iy_i.$$ The Euclidean dual of a linear code $\mathcal{C} \subset \mathbb{F}_q^n$ is defined by
    $$\mathcal{C}^\bot = \left\{\mathbf{y}\in \mathbb{F}_q^n : \mathbf{x}\cdot\mathbf{y} = 0 , \forall~\mathbf{x}\in \mathcal{C} \right\}.$$ Let $\mathcal{C}$ be a $\lambda$-constacyclic code of length $n$ over $\mathbb{F}_q$ with generator polynomial $g(x)$. Then, $\mathcal{C}^\bot$ is a $\lambda^{-1}$-constacyclic code of length $n$ over $\mathbb{F}_q$ with generator polynomial $h_R(x)$, where $h_R(x) = (x^n-\lambda^{-1})/g_R(x)$.

    Assume that $\gcd(n,q)=1$. Let $r$ be the order of $\lambda$ in $\mathbb{F}_q^*$ and $t$ be the multiplicative order of $q$ modulo $rn$. Then $rn \mid (q^t-1)$ and $n\mid (q^t-1)$. Let $\alpha$ be a primitive $rn$-th root of unity in $\mathbb{F}_{q^t}$ and $\eta = \alpha^r\in \mathbb{F}_{q^t}$. It is clear that $\eta$ is a primitive $n$-th root of unity. And  it follows that $\alpha\eta^i = \alpha^{1+ri}$ are all the roots of $x^n-\lambda$ for $0\leq i\leq n-1$. The $q$-cyclotomic coset $C_i^{(q,rn)}$ modulo $r n$ containing $i$ is defined by 
    $$C_i^{(q,rn)} = \left\{ iq^j \text{ mod } rn:0\leq j\leq \ell_i -1\right\},$$ where $\ell_i$ is the smallest positive integer such that $iq^{\ell_i}\equiv i \pmod {rn}$. 
    Obviously, the polynomial $m_{\alpha^i}(x) =\prod_{j\in C_i^{(q,n)}}(x-\alpha^i)$ must be in $\mathbb{F}_q[x]$ and is called the $minimal\ polynomial$ of $\alpha^i$ over $\mathbb{F}_q$. Denote $\Omega_{rn} = \{1+ri:0\leq i \leq n-1\}$. Let $\mathcal{C}$ be a $\lambda$-constacyclic code in $\mathbb{F}_q/\langle x^n-\lambda\rangle$ with generator polynomial $g(x)$. Then there exists $S\subseteq \Omega_{rn}$ such that $$g(x) = \prod_{i\in S}m_{\alpha^i}(x).$$ The set $$T=\bigcup_{i\in S} C_{i}^{(q,rn)}$$ is called the defining set of $\mathcal{C}$ with respect to $\alpha$.
     
    In this paper, we will employ simple-root cyclic codes to construct new MDS symbol-pair codes. The following lemmas are well-known lower bounds on the minimum distance of (consta)cyclic codes, which will be applied in our later proofs.
    \begin{lemma}[BCH bound for constacyclic codes \cite{Krishna1990}]
        \label{lem::BCH_Bound}
        Assume that $\gcd(n,q)=1$. Let $\mathcal{C}$ be a $\lambda$-constacyclic code over $\mathbb{F}_q$ of length $n$ with defining set $T$. Let {\rm ord}$(\lambda)=r$. If $\left\{1+ri: 0 \leq i \leq \delta -2 \right\} \subseteq T$, then the minimum Hamming distance of $\mathcal{C}$ is not less than $\delta$.
    \end{lemma}
    \begin{lemma}[Hartmann–Tzeng bound for cyclic codes \cite{Huffman2010}]\label{lem::HT_Bound}
    Let $\mathcal{C}$ be a cyclic code of length $n$ over $\mathbb{F}_q$ with defining set T. Let $A$ be a set of $\delta - 1 $ consecutive elements of T and $B =\left\{jb \pmod n | 0 \leq j \leq s\right\}$, where $\gcd(b, n) <\delta$. If $ A + B \subseteq T$ , then the minimum weight d of $\mathcal{C}$ satisfies $d \geq \delta + s$.
    \end{lemma}
The following lemma given in \cite{Chen2017} give the relationship between the Hamming distance and pair distance of constacyclic non-MDS codes.
    \begin{lemma}[\cite{Chen2017}]\label{lem::Chen2017}
        Let $\mathcal{C}$ be an $[n,k,d_H]$ constacyclic code over $\mathbb{F}_q$ with $2\leq d_H \leq n$. Then $d_P(\mathcal{C})\geq d_H+2$ if and only if C is not an MDS code.
    \end{lemma}

\subsection{The Decomposition of Cyclic Codes}
\label{sec:decomposition_of_cyclic_codes}

Suppose $m$ is an integer with $m\mid (q-1)$ and let $\{\zeta_0,\zeta_1,\cdots,\zeta_{m-1}\} \subseteq \mathbb{F}_q^*$ be the set of $m$-th root of unity. Define the map
\begin{equation*}
    \begin{aligned}
        \varphi: \frac{\mathbb{F}_q[x]}{\langle x^{mn}-1 \rangle} &\longrightarrow \frac{\mathbb{F}_q[x]}{\langle x^{n}-\zeta_0 \rangle} \oplus \frac{\mathbb{F}_q[x]}{\langle x^{n}-\zeta_1 \rangle} \oplus \cdots \oplus \frac{\mathbb{F}_q[x]}{\langle x^{n}-\zeta_{m-1} \rangle} \\
         p(x)\quad  &\longrightarrow \Big(\varphi_{\zeta_0}\big(p(x)\big) , \varphi_{\zeta_1}\big(p(x)\big),\cdots, \varphi_{\zeta_{m-1}}\big(p(x)\big) \Big)
    \end{aligned}
\end{equation*}
where $\varphi_{\zeta_{i}}\big(p(x)\big) = p(x) \ ( \text{mod } x^n-\zeta_i)$ for $0\leq i \leq m-1$. Suppose $p(x) = \sum_{i=0}^{mn-1}p_i x^i \in \mathbb{F}_q[x]/\langle x^{mn}-1\rangle $ and $\vec{q}_i =\varphi_{\zeta_i}\big(p(x)\big)$, then we have $VP = Q$ where
\begin{equation*}
    \begin{aligned}
        P=\begin{pmatrix}
        p_0 & p_1  & \cdots & p_{n-1} \\
        p_{n} & p_{n+1} & \cdots & p_{2n-1} \\
        \vdots & \vdots & \ddots & \vdots \\
        p_{(m-1)n} & p_{(m-1)(n+1)} & \cdots & p_{mn-1}
    \end{pmatrix}, \qquad 
    Q=\begin{pmatrix}
        q_0 & q_1  & \cdots & q_{n-1} \\
        q_{n} & q_{n+1} & \cdots & q_{2n-1} \\
        \vdots & \vdots & \ddots & \vdots \\
        q_{(m-1)n} & q_{(m-1)(n+1)} & \cdots & q_{mn-1}
    \end{pmatrix}
    \end{aligned}
\end{equation*}
 and 
\begin{equation*}
    V=\begin{pmatrix}
        1 & \zeta_0  & \cdots & \zeta_0^{n-1} \\
        1 & \zeta_1  & \cdots &  \zeta_1^{n-1} \\
        \vdots & \vdots & \ddots & \vdots \\
        1 & \zeta_{m-1}  & \cdots &  \zeta_{m-1}^{n-1} 
    \end{pmatrix}.
\end{equation*}
Since  $V$ is nonsingular, the map $\varphi$ is one-to-one and onto. In fact,  it is shown in \cite{Hughes2000}
that $\phi$ is a ring isomorphism: $$\mathbb{F}_q[x]/\langle x^{mn}-1\rangle \cong \bigoplus_{i=0}^{m-1} \mathbb{F}_q[x]/\langle x^{n}-\zeta_i\rangle.$$ Hence, a cyclic code of length $mn$ over $\mathbb{F}_q$ can decompose into constacyclic codes of length $n$ over $\mathbb{F}_q$ with respect to $\zeta_0, \zeta_1, \cdots, \zeta_{m-1}$. The following result is about the decomposition of a cyclic code. 

\begin{lemma}[\cite{Hughes2000}]
    \label{Lem:Decomposition}
Suppose $m \mid (q-1)$ and $\zeta_0,\zeta_1,\ldots,\zeta_{m-1} \in \mathbb{F}_q^*$ are all the $m$-th root of unity. Then a cyclic code of length $mn$ over $\mathbb{F}_q$ can be decomposed into constacyclic codes of length $n$ over $\mathbb{F}_q$ with respect to $\zeta_0,\zeta_1,\ldots,\zeta_{m-1}$.
\end{lemma}
In particular, let $q$ be odd and take $m=2$, $\zeta_0 = 1$, $\zeta_1 = -1$, then by Lemma \ref{Lem:Decomposition}, a cyclic code $\mathcal{C}$ over $\mathbb{F}_q$ of length $2n$ can be decomposed into a cyclic code $\mathcal{C}_1$ and a negacyclic code $\mathcal{C}_2$ of length $n$. Suppose $g(x) = g_1(x)g_2(x)$ is the generator polynomial of $\mathcal{C}$, where $g_1(x)$ and $g_2(x)$ are monic polynomials satisfying $g_1(x) \mid x^n-1$ and $g_2(x) \mid x^n+1$. Then $g_1(x)$ and $g_2(x)$ are generator polynomials of $\mathcal{C}_1$ and $\mathcal{C}_2$, respectively. In fact, we can provide a more intuitive relationship between $\mathcal{C}$ and $\mathcal{C}_1$, $\mathcal{C}_2$. By viewing the map $\varphi$ in matrix form, we can easily get that for any codeword $\mathbf
{c} = (c_0,c_1,\cdots,c_{2n-1})\in \mathcal{C}$, there exists $\mathbf{c}_1\in \mathcal{C}_1$ and $\mathbf{c}_2\in \mathcal{C}_2$ such that 
$$\begin{pmatrix}
    c_0 & c_1 & \cdots &  c_{n-1} \\
    c_{n} & c_{n+1} & \cdots &  c_{2n-1}
\end{pmatrix}
= V^{-1} \binom{\mathbf{c}_1}{\mathbf{c}_2}=2^{-1}\begin{pmatrix}
    1 & 1 \\
    1 & -1
\end{pmatrix}\binom{\mathbf{c}_1}{\mathbf{c}_2}.$$
 Since $q$ is odd, we can obtain $\mathbf{c}_1^\prime = 2^{-1}\mathbf{c}_1 \in \mathcal{C}_1$ and $\mathbf{c}_2^\prime = 2^{-1}\mathbf{c}_2 \in \mathcal{C}_2$. Thus we can obtain that $$\mathcal{C} = \{(\mathbf{u}+\mathbf{v},\mathbf{u}-\mathbf{v}): \mathbf{u}\in \mathcal{C}_1, \mathbf{v} \in \mathcal{C}_{2} \} \triangleq \mathcal{C}_1 \curlyvee \mathcal{C}_2.$$ 
 This decomposition of cyclic codes will be used in our construction of MDS symbol-pair codes of minimum pair distance $d_P=9$.

\section{Constructions of MDS symbol-pair codes}\label{sec3}
In this section, we will construct three classes of MDS symbol-pair codes with minimum pair distance of 7, 8, and 9, respectively.

\subsection{MDS symbol-pair codes with $d_P=7$ and $n=4q+4$}
In this subsection, we suppose $q$ is a prime power with $q\equiv 1\pmod 4$, we present a class of MDS symbol-pair codes of length $n=4q+4$ and minimum symbol-pair distance $d_P=7$.

Note that $n = 4q+4 \mid (q^2-1)$, so we let $\xi\in \mathbb{F}_{q^2}$ be a primitive $n$-th root of unity. Define $$g(x) = (x-1)(x+1)(x-\xi)(x-\xi^q)(x-\xi^{q+1}).$$
Then $g(x) \in \mathbb{F}_q[x]$ and $g(x) \mid x^n-1$. Let $\mathcal{C}=\langle g(x) \rangle$ be the cyclic code over $\mathbb{F}_q$ of length $n$ with generator polynomial  $g(x)$.

The goal of this subsection is to prove that $\mathcal{C}$ is an MDS symbol-pair code with $d_P=7$. Firstly, we determine the minimum Hamming distance of $\mathcal{C}$.

    \begin{lemma}\label{constr1::hammingdistance}
      The minimum Hamming distance $d_H(\mathcal{C})$ of $\mathcal{C}$ is equal to $4$.
    \end{lemma}
    \begin{proof}
        Note that $1,\xi,\xi^{q},\xi^{q+1}$ are roots of $g(x)$ and $\gcd(q,n) = 1$. By Lemma \ref{lem::HT_Bound}, we have $d_H(\mathcal{C}) \geq 3+1 = 4.$  Since $\xi$ is a primitive $n$-th root of unity, we have $(\xi^{q+1})^4=1$ and $\xi^{2q+2}=-1$. Thus $(x-1)(x+1)(x-\xi^{q+1}) \mid (x^4-1)$ and $(x-\xi)(x-\xi^q) \mid (x^{2q+2}+1)$. Therefore, $g(x) \mid (x^4-1)(x^{2q+2}+1)$, which implies that $(x^4-1)(x^{2q+2}+1) \in \mathcal{C}$. Thus $d_H(\mathcal{C}) \leq \Wt_H\big((x^4-1)(x^{2q+2}+1)\big) = 4$.  Hence $d_H(\mathcal{C}) = 4$.
    \end{proof}

    For convenience, we denote by the symbol $\star$ an element in $\mathbb{F}_q^*$ and by $\mathbf{0}_t$ an all-zero vector of length $t$. 

    \begin{lemma}
        \label{lemma6}
       There exists no codeword of the form $\left(\star,\star,\star,\mathbf{0}_r,\star,\mathbf{0}_{n-r-4}\right)$ in $\mathcal{C}$ for any $1\leq r \leq n-5$.
    \end{lemma}
    \begin{proof}
       We prove the lemma by contradiction. W.l.o.g., we suppose $c(x) = 1+c_1 x + c_2 x^2 + c_{\ell}x^{\ell} \in \mathcal{C}$, where $c_1,c_2,c_\ell \in \mathbb{F}^*_q$ and $4\leq\ell=r+3 \leq n-2$. Then $g(x) \mid c(x)$ $\Longrightarrow$ $c(1) = c(-1) = c(\xi) = c(\xi^{q+1}) =0$. So we have the following system
        \begin{equation}
            \label{Eq::Equation1}
            \begin{cases}
                1 + c_1 + c_2 + c_\ell &= 0,   \\
                1 - c_1 + c_2 + (-1)^\ell c_\ell&= 0, \\
                1 + c_1 \xi + c_2\xi^2 + c_\ell \xi^\ell &= 0, \\
                1 + c_1 \xi^{q+1} + c_2\xi^{2q+2} + c_\ell \xi^{(q+1)\ell} &= 0. 
        \end{cases}
        \end{equation}
        \begin{itemize}
            \item If $\ell$ is even, then from the first and second equations in (\ref{Eq::Equation1}), we can derive that $c_1 = 0$, which contradicts the assumption that $c_1 \in \mathbb{F}_q^*$.
            \item If $\ell$ is odd, then from the first and second equations in (\ref{Eq::Equation1}), we obtain that $c_1 = -c_\ell$, $c_2 = -1$. Then we have the following system:
            \begin{equation*}
                \begin{cases}
                    1 - \xi^2 + c_1 (\xi  - \xi^\ell) &= 0,\\
                    1  - \xi^{2 q+2}  + c_1 \xi^{q+1}(1-  \xi^{(q+1)(\ell-1)})  &= 0.
                \end{cases}
            \end{equation*}
            When $\ell \equiv 1 \pmod{4}$, we have $\xi^{(q+1)(\ell-1)} = 1$. Thus, by the second equation above, $\xi^{2q+2} = 1$, which contradicts the fact that $\xi$ is a primitive $(4q+4)$-th root of unity.
            
            When $\ell  \equiv 3 \pmod 4$, by solving the system, we get
            $$\begin{cases}
                c_1 = \frac{1-\xi^2}{\xi^\ell - \xi},\\
                        c_1 =  \frac{1 - \xi^{2q+2}}{\xi^{(q+1)\ell}-\xi^{q+1}} = -\frac{1}{\xi^{q+1}}.
            \end{cases}$$
            Thus $\frac{\xi^2-1}{\xi^\ell - \xi} = \frac{1}{\xi^{q+1}}$, i.e.,
            \begin{align}
                \xi^l = \xi^{q+1}(\xi^2-1) + \xi.
                \label{Eq::Equation3}
            \end{align}

            Since $\frac{\xi^2-1}{\xi^\ell - \xi} = \frac{1}{\xi^{q+1}}\in \mathbb{F}_q^*$, we have
                    \begin{equation*}
                        \begin{aligned}
                            (\frac{\xi^2-1}{\xi^\ell - \xi})^q 
                            &= \frac{\xi^{2q}-1}{\xi^{q\ell} - \xi^q} = \frac{-\xi^{-2}-1}{\xi^q(-\xi^{-(\ell-1)}-1) } \\
                            &= \frac{(\xi^{-2}+1)\xi^{\ell-1}}{\xi^q(1+\xi^{\ell-1}) } = \frac{1}{\xi^{q+1}}.
                        \end{aligned}
                    \end{equation*}
     Thus
                    \begin{align}
                        (\xi^{-2}+1)\xi^{\ell}  = 1+\xi^{\ell-1}.
                        \label{Eq::Equation4}
                    \end{align} Combining (\ref{Eq::Equation3}) with (\ref{Eq::Equation4}), we get 
                    \begin{equation*}
                        \begin{aligned}
                            (\xi^{-2}+1)\left[\xi^{q+1}(\xi^2-1) + \xi\right]  = 1+\xi^{q}(\xi^2-1) + 1.
                        \end{aligned}
                    \end{equation*}
                    That is, 
                    \begin{align*}
                        (\xi - 1)(\xi^{q+2}+\xi^{q-1}+1-\xi^{-1})=0.
                    \end{align*} 
                    Hence, $\xi^{q+2}+\xi^{q-1}+1-\xi^{-1} = 0  \Longrightarrow \xi^{q+1} = \frac{\xi^{-1}-1}{\xi+\xi^{-2}} = \frac{\xi-\xi^2}{\xi^3+1}$. Note that $\xi^{2q+2} =-1$, so $\xi^{2q+2} = (\frac{\xi-\xi^2}{\xi^3+1})^2 = -1$.
                    Then we have 
                    $$\xi^6+\xi^4 + \xi^2 + 1 =(\xi^2+1)(\xi^4+1)= 0,$$
                    which implies that $\xi^8=1$. It contradicts the condition that $\xi$ is a primitive $n$-th root of unity.
        \end{itemize}
        In summary, there does not exist such a codeword in $\mathcal{C}$.
    \end{proof}

    \begin{lemma}
        \label{lemma7}
         There exists no codeword of the form $\left(\star,\star,\mathbf{0}_r,\star,\star,\mathbf{0}_{n-r-4}\right)$ in $\mathcal{C}$ for any $1\leq r \leq n-5$.
    \end{lemma}
    \begin{proof}
      We prove the lemma by contradiction. W.l.o.g., we assume that $c(x) = 1+c_1x+ c_\ell x^\ell + c_{\ell+1} x^{\ell+1} \in \mathcal{C}$, where $c_1,c_\ell,c_{\ell+1} \in \mathbb{F}^*_q$ and $3\leq\ell=r+2 \leq n-3$. Since $1,-1,\xi,\xi^{q+1}$ are four roots of $g(x)$, we immediately have $c(-1) = c(1) = c(\xi) = c(\xi^{q+1}) = 0.$ Hence, we have the following system:
        \begin{equation}
            \label{Eq::Equation5}
            \begin{cases}
                1 + c_1 + c_\ell + c_{\ell+1} &= 0,\\
                1 - c_1 + (-1)^{\ell}c_\ell + (-1)^{\ell+1} c_{\ell+1}&= 0,\\
                1 + c_1 \xi  + c_\ell \xi^\ell + c_{\ell+1}\xi^{\ell+1}&= 0, \\
                1 + c_1 \xi^{q+1} + c_\ell \xi^{(q+1)\ell} + c_{\ell+1}\xi^{(q+1)(\ell+1)} &= 0.
            \end{cases}
        \end{equation}
        \begin{itemize}
            \item When $\ell$ is even, from the first and second equations of (\ref{Eq::Equation5}), we can derive $c_1  = -c_{\ell+1}$ and $c_{\ell}=-1$. From the third equation of (\ref{Eq::Equation5}) , we obtain  $$1 + c_1 \xi  -  \xi^\ell - c_{1}\xi^{\ell+1}= (1-\xi^{\ell})(1-c_1\xi) =0.$$
            Note that $1-\xi^{\ell} \neq 0$ for $3\leq \ell \leq n-3$, then $c_1 = \frac{1}{\xi} \not \in \mathbb{F}_q$, which is a contradiction.
            \item When $\ell$ is odd, from the first and second equations of (\ref{Eq::Equation5}), we can obtain $c_1  = -c_\ell$ and $c_{\ell+1}=-1$. Then we have the following system:
            \begin{numcases}{}
                1 - \xi^{\ell+1} + c_1 (\xi  -  \xi^\ell ) &= 0, \nonumber\\ 
                1 - \xi^{(q+1)(\ell+1)}  + c_1 \xi^{q+1}(1 -  \xi^{(q+1)(\ell-1)}) &= 0. \label{Eq::Equation6} 
            \end{numcases}
            If $\ell \equiv 1 \pmod 4$, we have $\xi^{(q+1)(\ell-1)} = 1$ and $\xi^{(q+1)(\ell+1)} =\xi^{2(q+1)} =  -1$. By Equation (\ref{Eq::Equation6}), we can obtain $ 2 = 0$, which is a contradiction.
            \newline If $\ell \equiv 3 \pmod 4$, we have $\xi^{(q+1)(\ell+1)} = 1$ and $\xi^{(q+1)(\ell-1)} =\xi^{2(q+1)} = -1$. By Equation (\ref{Eq::Equation6}), we can obtain $ 2c_1\xi^{q+1} = 0 \Longrightarrow c_1=0$, which contradicts the assumption that $c_1\in \mathbb{F}_q^*$.
        \end{itemize}
        In summary, there does not exist such a codeword in $\mathcal{C}$.
    \end{proof}
    Combining Lemmas \ref{lemma6} and \ref{lemma7}, we can construct a class of MDS symbol-pair codes with $d_P = 7$.
    \begin{theorem}
        \label{Thm:theorem1}
        Suppose $q$ is a prime power with $q \equiv 1\pmod 4$ and $n=4q+4$. Let $\mathcal{C}$ be the cyclic code with generator polynomial $g(x) = (x-1)(x+1)(x-\xi)(x-\xi^q)(x-\xi^{q+1})$. Then $\mathcal{C}$ is an $(n=4q+4, d_P=7)_q$ MDS symbol-pair code.
    \end{theorem}
    \begin{proof}
       Note that $\dim(\mathcal{C})=n-\deg(g(x))=4q-1$ and $d_H(\mathcal{C})=4$ by Lemma \ref{constr1::hammingdistance}, thus $\mathcal{C}$ is not an MDS code. Hence we can obtain $d_P(\mathcal{C})\geq 6$ by Lemma \ref{lem::Chen2017}. By the Singleton-type bound (\ref{singleton_type_bound}), we have $d_P(\mathcal{C}) \leq 7$, thus it is sufficient to prove that $d_P(\mathcal{C}) \neq 6$, equivalently,  there exists no codeword with minimum pair weight $6$. By contradiction, suppose that there is a codeword $\mathbf{c}\in\mathcal{C}$ with pair weight 6, then its certain cyclic shift must be one of the following forms: 
         \begin{gather*}
(\star,\star,\star,\star,\star,\mathbf{0}_{n-5}),\\
(\star,\star,\star,\mathbf{0}_{r},\star,\mathbf{0}_{n-r-4}) \textnormal{ with } 1\leq r\leq n-5, \\
            (\star,\star,\mathbf{0}_{r},\star,\star,\mathbf{0}_{n-r-4})  \textnormal{ with } 1\leq r\leq n-5.
        \end{gather*}
        If $\mathbf{c}$ is of the first form, then the corresponding polynomial $c(x)$ has degree 4, while the generator polynomial $g(x)$ has degree 5, which is impossible. By Lemmas \ref{lemma6} and \ref{lemma7}, $\mathbf{c}$ can not be of the other two forms. Hence, we have proved that $\mathcal{C}$ is a $(4q+4,7)_q$ MDS symbol-pair code.
    \end{proof}
    \begin{remark}
        There are several constructions of MDS symbol-pair codes with $d_P = 7$. When $q = p$ is a prime, Ma et al. (\cite{Ma2022}) proposed a family of $(n = 5p, d_P = 7)_q$-MDS symbol-pair codes, which have the longest code length known so far. When $q$ is a prime power, Kai et al. (\cite{Kai2024}) constructed another family of $(n = 4q - 4, d_P = 7)_q$-MDS symbol-pair codes for $q \equiv 1 \pmod{4}$. Thus, when $q$ is not a prime, our MDS symbol-pair codes achieve the longest code length with minimum pair distance 7 to date.
    \end{remark}

\subsection{MDS symbol-pair codes with $d_P=8$ and $n=4q-4$}
In this subsection, we suppose $n = 4q-4$ and $q$ is a prime power satisfying  $q\equiv 3 \pmod 4$. 

Note that $n = 4q-4 \mid (q^2-1)$, so we let $\xi\in \mathbb{F}_{q^2}$ be a primitive $n$-th root of unity. Define $$g(x) = (x-1)(x+1)(x-\xi)(x-\xi^q)(x-\xi^{2})(x-\xi^{2q}).$$
Then $g(x) \in \mathbb{F}_q[x]$ and $g(x) \mid x^n-1$. Let $\mathcal{C}=\langle g(x) \rangle$ be the cyclic code over $\mathbb{F}_q$ of length $n$ with generator polynomial $g(x)$.

The goal of this subsection is to prove that $\mathcal{C}$ is an MDS symbol-pair code with $d_P=8$. Firstly, we  determine the minimum Hamming distance of $\mathcal{C}$.
    \begin{lemma}
        Denote $d_H(\mathcal{C})$ as the minimum Hamming distance of $\mathcal{C}$, then
        \begin{itemize}
            \item[(1)]   $d_H(\mathcal{C}) = 6$ when $q=3$;
            \item[(2)] $d_H(\mathcal{C}) = 4$ when $q\neq 3$.
        \end{itemize}
    \end{lemma}
    \begin{proof}
        (1) When $q = 3$, it is directly to verify that $\mathcal{C}$ is a $[8,2,6]$ cyclic code.
            
        (2) When $q\neq 3$, note that $1,\xi,\xi^{2}$ are roots of $g(x)$. By the BCH bound, $d_H(\mathcal{C}) \geq  4.$ Since $\xi$ is a primitive $n$-th root of unity, we have $(\xi^{2})^{2q-2}=1$, $\xi^{q+1}\in \mathbb{F}_q^*$, and $\xi^{2q-2}=-1$.
        Thus $(x-1)(x+1)(x-\xi^{2})(x-\xi^{2q}) \mid x^{2q-2}-1$ and $(x-\xi)(x-\xi^q) \mid \xi^{q+1}x^{q-3}+1$. Therefore, $g(x) \mid (x^{2q-2}-1)(\xi^{q+1}x^{q-3}+1)$, which implies that $(x^{2q-2}-1)(\xi^{q+1}x^{q-3}+1)\in \mathcal{C}$. Thus $d_H(\mathcal{C})\leq \Wt_H\big((x^{2q-2}-1)(\xi^{q+1}x^{q-3}+1)\big) = 4$. Hence $d_H(\mathcal{C}) = 4$. The desired result follows.
    \end{proof}
    In \cite{Kai2024}, Kai et al. constructed an MDS symbol-pair code with $d_P=7$.
    \begin{lemma}[{\cite[Theorem 3.14]{Kai2024}}]
        \label{lemma9}
         Let $\mathcal{C}^{\prime}$ be the cyclic code with generator polynomial $g'(x) = (x-1)(x-\xi)(x-\xi^q)(x-\xi^2)(x-\xi^{2q})$. Then $\mathcal{C}^{\prime}$ is an $(n=4q-4,d_P=7)_q$ MDS symbol-pair code.
    \end{lemma}
   Note that $g'(x) \mid g(x)$, thus $\mathcal{C}$ is a subcode of $\mathcal{C}^{\prime}$. So we have $d_P(\mathcal{C})\geq d_P(\mathcal{C}^{\prime}) \geq 7$. To prove $d_P(\mathcal{C})=8$, it suffices to show that there exists no codeword with pair weight $7$. Therefore, in the following lemmas, we will respectively prove that $\mathcal{C}$ does not contain a codeword of the following forms:
    \begin{equation*}
        \begin{gathered}
\left(\star,\star,\star,\star,\mathbf{0}_r,\star,\mathbf{0}_{n-r-5}\right) \text{ where }1\leq r \leq n-6,\\
\left(\star,\star,\star,\mathbf{0}_r,\star,\star,\mathbf{0}_{n-r-5}\right) \text{ where }1\leq r \leq n-6,\\
\left(\star,\star,\mathbf{0}_r,\star,\mathbf{0}_s,\star,\mathbf{0}_{n-r-s-4}\right) \text{ where } 1\leq r,s\leq n-5 \text{ and }2\leq r+s \leq n-5.
        \end{gathered}
    \end{equation*}

    \begin{lemma}
        \label{lemma10}
         There exists no codeword of the form $\left(\star,\star,\star,\star,\mathbf{0}_r,\star,\mathbf{0}_{n-r-5}\right)$ in $\mathcal{C}$ for any $1\leq r \leq n-6$.
    \end{lemma}
    \begin{proof}
         We prove the lemma by contradiction. W.l.o.g., we suppose $c(x) = 1+c_1 x + c_2 x^2 + c_3 x^3 + c_{\ell}x^{\ell}$, where $c_1,c_2,c_\ell \in \mathbb{F}^*_q$ and $5\leq\ell \leq n-2$. Note that $\xi^{2q} = -\xi^{2}$ and $g(x) \mid c(x)$. This means that $c(1) = c(-1) = c(\xi) = c(\xi^2) = c(-\xi^2) =0$. So we have the following system
        \begin{equation}
            \label{Eq::Equation7}
            \begin{cases}
                1 + c_1 + c_2 + c_3 + c_\ell &= 0, \\
                1 - c_1 + c_2 - c_3 + (-1)^\ell c_\ell&= 0, \\
                1 + c_1 \xi + c_2\xi^2  + c_3 \xi^3+ c_\ell \xi^\ell &= 0, \\
                1 + c_1 \xi^{2} + c_2\xi^{4} + c_3 \xi^{6} + c_\ell \xi^{2\ell} &= 0,  \\
                1 - c_1 \xi^{2} + c_2\xi^{4} - c_3 \xi^{6} + (-1)^\ell c_\ell \xi^{2\ell} &= 0.
            \end{cases}
        \end{equation}
        \begin{itemize}
            \item If $\ell$ is odd, from the lines 1,2,4, and 5 of  system (\ref{Eq::Equation7}), we can derive that
                \begin{equation*}
                    \begin{cases}
                        1 + c_2  &= 0,  \\
                        c_1 +  c_3 + c_\ell&= 0, \\
                        1 + c_2\xi^{4}  &= 0, \\
                        c_1 \xi^{2} + c_3 \xi^{6} +  c_\ell \xi^{2\ell} &= 0. 
                    \end{cases}
                    \Longrightarrow \xi^4=1.
                \end{equation*}
                It contradicts the fact that $\xi$ is a primitive $(4q-4)$-th root of unity.
            \item If $\ell$ is even, from the first two and last two equations of (\ref{Eq::Equation7}), we can obtain that
                \begin{equation*}
                    \begin{cases}
                        1 + c_2 + c_\ell  &= 0,  \\
                        c_1 +  c_3 &= 0, \\
                        1 + c_2\xi^{4} + c_\ell\xi^{2\ell} &= 0, \\
                        c_1 \xi^{2} + c_3 \xi^{6} &= 0. 
                    \end{cases}
                    \Longrightarrow c_1 \xi^2(1-\xi^4)=0.
                \end{equation*}
                Since $c_1\in \mathbb{F}_q^*$ and $\xi^2\neq 0$, we can get $1-\xi^4=0$, which is a contradiction.
        \end{itemize}
        In summary, there does not exist such a codeword in $\mathcal{C}$.
    \end{proof}

    \begin{lemma}
        \label{lemma11}
        There exists no codeword of the form $\left(\star,\star,\star,\mathbf{0}_r,\star,\star,\mathbf{0}_{n-r-5}\right)$ in $\mathcal{C}$ for any $1\leq r\leq n-6$.
    \end{lemma}
    \begin{proof}
        We prove the lemma by contradiction. W.l.o.g., we can assume that $c(x) = 1+c_1x + c_2 x^2+ c_\ell x^\ell + c_{\ell+1} x^{\ell+1}$, where $c_1,c_\ell,c_{\ell+1} \in \mathbb{F}^*_q$ and $4\leq\ell \leq n-3$. Since $1,-1,\xi,\xi^{2},-\xi^{2}$ are five roots of $g(x)$, we immediately have $c(-1) = c(1) = c(\xi) = c(\xi^{2}) = c(-\xi^{2}) =  0.$ Hence, we have the following system:
        \begin{equation*}
            \begin{cases}
                1 + c_1 + c_2 + c_\ell + c_{\ell+1} &= 0,\\
                1 - c_1 + c_2 + (-1)^{\ell}c_\ell + (-1)^{\ell+1} c_{\ell+1}&= 0,\\
                1 + c_1 \xi + c_2 \xi^2 + c_\ell \xi^\ell + c_{\ell+1}\xi^{\ell+1}&= 0,\\
                1 + c_1 \xi^{2} + c_2 \xi^{4} + c_\ell \xi^{2\ell} + c_{\ell+1}\xi^{2(\ell+1)} &= 0,\\
                1 - c_1 \xi^{2} + c_2 \xi^{4} + (-1)^{\ell} c_\ell \xi^{2\ell} + (-1)^{\ell+1}c_{\ell+1}\xi^{2(\ell+1)} &= 0.
            \end{cases}
        \end{equation*}
        \begin{itemize}
            \item If $\ell$ is even, we can obtain that 
                \begin{equation}
                    \label{Eq:Equation8}
                    \begin{cases}
                        1 + c_2 + c_\ell  &= 0,  \\
                        c_1 + c_{\ell+1} &= 0, \\
                        1 + c_1 \xi + c_2 \xi^2 + c_\ell \xi^\ell + c_{\ell+1}\xi^{\ell+1}&= 0,\\
                        1  + c_2\xi^4 + c_\ell \xi^{2\ell}  &= 0,  \\
                        c_1\xi^2 + c_{\ell+1}\xi^{2\ell+2} &= 0.
                    \end{cases}
                \end{equation}
                By the second and the last equations of (\ref{Eq:Equation8}), we derive that $c_1\xi^2(1-\xi^{2\ell}) = 0.$
                Since $c_1\in \mathbb{F}_q^*$ and $\xi^2\neq 0$, we can get $\xi^{2\ell}=1$. Then from the first and fourth equations of (\ref{Eq:Equation8}), we obtain 
                \begin{equation*}
                    \begin{cases}
                        1 + c_2 + c_\ell  &= 0,\\
                        1  + c_2\xi^4 + c_\ell &= 0.
                    \end{cases}
                    \Longrightarrow c_2(\xi^4-1)=0.
                \end{equation*}
                Since $c_2\in\mathbb{F}_q^*$, we have $\xi^4=1$, which contradicts the fact that $\xi$ is a primitive $(4q-4)$-th root of unity.
            \item If $\ell$ is odd, we can get 
                \begin{equation}
                    \label{Eq:Equation9}
                    \begin{cases}
                        1 + c_2 + c_{\ell+1}  &= 0,  \\
                        c_1 + c_{\ell} &= 0, \\
                        1 + c_1 \xi + c_2 \xi^2 + c_\ell \xi^\ell + c_{\ell+1}\xi^{\ell+1}&= 0,\\
                        1  + c_2\xi^4 + c_{\ell+1} \xi^{2\ell+2}  &= 0,  \\
                        c_1\xi^2 + c_{\ell}\xi^{2\ell} &= 0.
                    \end{cases}
                    \Longrightarrow 
                    c_1\xi^2(1-\xi^{2\ell-2}) = 0.
                \end{equation}
                By the second and last equations of (\ref{Eq:Equation9}), we derive $c_1\xi^2(1-\xi^{2\ell-2}) = 0.$ Since $c_1\in \mathbb{F}_q^*$ and $\xi^2\neq 0$, we can get $\xi^{2\ell-2}=1$. Similarly, from the first and the fourth equations of (\ref{Eq:Equation9}), we obtain 
                \begin{equation*}
                    \begin{cases}
                        1 + c_2 + c_{\ell+1}  &= 0,\\
                        1  + c_2\xi^4 + c_{\ell+1}\xi^4 &= 0.
                    \end{cases}
                    \Longrightarrow \xi^4-1=0.
                \end{equation*}
                It contradicts the fact that $\xi^n =  1$.
        \end{itemize}
        In summary, there does not exist such a codeword in $\mathcal{C}$.
    \end{proof}

    \begin{lemma}
        \label{lemma12}
        There exists no codeword of the form $\left(\star,\star,\mathbf{0}_r,\star,\mathbf{0}_s,\star,\mathbf{0}_{n-r-s-4}\right)$ in $\mathcal{C}$ for $1\leq r,s\leq n-5$ and $2\leq r+s \leq n-5$. 
    \end{lemma}

    \begin{proof}
        We prove the lemma by contradiction.  W.l.o.g., we can assume that $c(x) = 1+c_1x + c_{\ell_1} x^{\ell_1} + c_{\ell_2} x^{\ell_2}$, where $c_1,c_{\ell_1},c_{\ell_2} \in \mathbb{F}^*_q$ , $3\leq \ell_1, \ell_2 \leq n-3$ and $\ell_1 -\ell_2 \neq \pm 1$. Since $1,-1,\xi,\xi^{2}, \xi^{2q}=-\xi^2$ are five roots of $g(x)$, we immediately have $c(-1) = c(1) = c(\xi) = c(\xi^{2}) = c(-\xi^{2}) = 0$. Hence, we have the following system:
        \begin{equation}
            \label{Eq::Equation10}
             \begin{cases}
                1 + c_1 + c_{\ell_1} + c_{\ell_2} &= 0,  \\
                1 - c_1 + (-1)^{\ell_1}c_{\ell_1} + (-1)^{\ell_2} c_{\ell_2}&= 0, \\
                1 + c_1 \xi  + c_{\ell_1} \xi^{\ell_1} + c_{\ell_2}\xi^{\ell_2}&= 0,\\
                1 + c_1 \xi^{2} + c_{\ell_1} \xi^{2\ell_1} + c_{\ell_2}\xi^{2\ell_2} &= 0,\\ 
                1 - c_1 \xi^{2} + (-1)^{\ell_1} c_{\ell_1} \xi^{2\ell_1} + (-1)^{\ell_2}c_{\ell_2}\xi^{2\ell_2} &= 0.
             \end{cases}
        \end{equation}
        \begin{itemize}
            \item If $\ell_1$ and $\ell_2$ are both odd or even. From the first and second equations of (\ref{Eq::Equation10}), we can get $2=0$ or $c_1 = 0$, which is a contradiction.
            \item If $\ell_1$ is odd and $\ell_2$ is even, we have
            \begin{equation*}
                \begin{cases}
                    1 + c_{\ell_2} &= 0, \\
                    c_1 + c_{\ell_1} &=  0,\\
                    1 + c_{\ell_2}\xi^{2\ell_2} &= 0, \\
                    c_1\xi^{2} + c_{\ell_1}\xi^{2\ell_1} &= 0. \\
                \end{cases}
                \Longrightarrow 
                \begin{cases}
                    \xi^{2\ell_2} = 1,\\
                    c_1\xi^2(1-\xi^{2\ell_1-2}) = 0.\\
                \end{cases}
            \end{equation*}
            Since $c_1\in \mathbb{F}_q^*$ and $\xi^2\neq 0$, we can get $\xi^{2\ell_1-2}=1$. That means that  $n \mid 2\ell_2$ and $n \mid 2\ell_1-2$. Note that $3\leq \ell_1, \ell_2 \leq n-3$. We obtain $\ell_1-1 =  \ell_2 = \frac{n}{2}$. It contradicts the assumption that $\ell_1-\ell_2 \neq \pm 1$.
            \item If $\ell_1$ is even and $\ell_2$ is odd, we can similarly deduce a contradiction.
        \end{itemize}
        In summary, there does not exist such a codeword in $\mathcal{C}$. This completes the proof.
    \end{proof}
    Combining Lemmas \ref{lemma10}, \ref{lemma11} and \ref{lemma12}, we can construct a class of MDS symbol-pair codes with $d_P = 8$.

\begin{theorem}
    \label{Thm:theorem2}
    Suppose $q$ is a prime power with $q\equiv 3\pmod 4$ and $n = 4q - 4$. Let $\mathcal{C}$ be the cyclic code in $\mathbb{F}_q\left[x\right]/\langle x^n-1\rangle$ with generator polynomial $g(x) = (x-1)(x+1)(x-\xi)(x-\xi^q)(x-\xi^{2})(x-\xi^{2q})$. Then $\mathcal{C}$ is an $(n=4q-4,d_P=8)_q$ MDS symbol-pair code.
\end{theorem}
\begin{proof}
   When $q=3$, then $\mathcal{C}$ is a $[8,2,6]$ cyclic code. Since $\mathcal{C}$ is not an MDS code, we can obtain $d_P(\mathcal{C})\geq 8$ by Lemma \ref{lem::Chen2017}. On the other hand, the Singleton-type bound (\ref{singleton_type_bound}) implies that $d_P\leq 8$. Thus $d_P = 8$ and $\mathcal{C}$ is an MDS symbol-pair code.
    
 When $q > 3$, then $\mathcal{C}$ is an $[n,n-6,4]$ cyclic code. Since the code $\mathcal{C}$ is a subcode of $\mathcal{C}'$ given in Lemma \ref{lemma9}, we can obtain $d_P(\mathcal{C})\geq 7$. By the Singleton-type bound (\ref{singleton_type_bound}), it is sufficient to prove that $d_P(\mathcal{C}) \neq 7$, equivalently,  there exists no codeword with minimum pair weight $7$. By contradiction, suppose that there is a codeword $\mathbf{c}\in\mathcal{C}$ with pair weight 7, then its certain cyclic shift must be one of the forms 
    \begin{gather*}
(\star,\star,\star,\star,\star,\star,\mathbf{0}_{n-6}),  \\
(\star,\star,\star,\star,\mathbf{0}_{r},\star,\mathbf{0}_{n-r-5}) \textnormal{ with }1\leq r\leq n-6, \\
        (\star,\star,\star,\mathbf{0}_{r},\star,\star,\mathbf{0}_{n-r-5}) \textnormal{ with }1\leq r\leq n-6,\\
        (\star,\star,\mathbf{0}_{r},\star,\mathbf{0}_{s},\star,\mathbf{0}_{n-r-s-4}) \textnormal{ with }1\leq r,s \leq n-5, 2\leq r+s \leq n-5. 
    \end{gather*}
    If $\mathbf{c}$ is of the first form, then the corresponding polynomial $c(x)$ has degree 5, while the generator polynomial $g(x)$ has degree 6, which leads to a contradiction. By Lemmas \ref{lemma10}, \ref{lemma11}, and \ref{lemma12}, no codeword in $\mathcal{C}$ can have the last three forms. Hence, we conclude that $\mathcal{C}$ is an $(n=4q-4,d_P=8)_q$ MDS symbol-pair code.
\end{proof}
\begin{remark}
    There are some constructions of MDS symbol-pair codes with $d_P = 8$. When $q = p$ is a prime, Ma et al. (\cite{Ma2022}) proposed a family of $(n = 4p, d_P = 8)_q$-MDS symbol-pair codes with $5 \mid (q-1)$, which have the longest code length. When $q=2^t$, Chee et al. \cite{Chee2013} constructed another family of $(n = 2q+4, d_P = 8)_q$-MDS symbol-pair codes. Thus, when $q$ is not a prime, our MDS symbol-pair codes achieve the longest code length with minimum pair distance $8$ to date.
\end{remark}

\subsection{MDS symbol-pair codes with $d_P = 9$ and $n=2q+2$}
  In \cite{Kai2024}, Kai et al. present two new families of MDS symbol-pair codes with $d_P=7$ by using the decomposition of cyclic codes. Inspired by their method, in this subsection, we will consider the construction of MDS symbol-pair codes with $d_P=9$ and code length $n=2q+2$, where $q$ is an odd prime power.
  
  Note that $n = 2q+2 \mid (q^2-1)$, so we let $\xi\in \mathbb{F}_{q^2}$ be a primitive $n$-th root of unity. Then $C_{-1}^{(q,n)} = \{-1,-q\}$, $C_0^{(q,n)} = \{0\}$, $C_1^{(q,n)} = \{1,q\}$, and $C_2^{(q,n)} =\{2,2q\}$. Define $$g(x)=m_{\xi^{-1}}\left(x\right)m_{\xi^{0}}\left(x\right)m_{\xi^{1}}\left(x\right)m_{\xi^{2}}\left(x\right).$$
    Then $g(x) \in \mathbb{F}_q[x]$ and $g(x) \mid x^n-1$. Let $\mathcal{C}=\langle g(x) \rangle$ be the cyclic code over $\mathbb{F}_q$ of length $n$ with generator polynomial $g(x)$.

    The goal of this subsection is to prove that $\mathcal{C}$ is an MDS symbol-pair code with $d_P=9$. Firstly, we  determine the minimum Hamming distance of $\mathcal{C}$.

    \begin{lemma}
        \label{constr3::hammingdistance}
        Denote $d_H(\mathcal{C})$ as the minimum Hamming distance of $\mathcal{C}$, then 
        \begin{itemize}
            \item [(1)] $d_H\left(\mathcal{C}\right)= 8$ when $q=3$;
            \item [(2)] $d_H\left(\mathcal{C}\right)= 6$ when $q\geq 5$.
        \end{itemize}
    \end{lemma}
    \begin{proof}
        Note that the defining set of $\mathcal{C}$ is $T=\{-q,-1,0,1,2,q,2q \pmod n\}$. 
        \begin{itemize}
            \item [(1)] When $q=3$, then $T=\{0,1,2,3,5,6,7\}$. It is clear that $\mathcal{C}$ is a $[8,1,8]$ repetition code.
            \item [(2)] When $q\geq 5$, note that $2q\equiv -2 \pmod n$, the defining set of $\mathcal{C}$ is $T=\{-q,-1,0,1,2,q,-2\pmod n\}$.
            It follows that $d_H(\mathcal{C})\geq 6$ by the BCH bound. Since $\xi$ is a primitive $n$-th root of unity, we have $\xi^{2q+2}=1$ and $\xi^{q+1} =-1$. Thus, $m_{\xi^{-1}}\left(x\right)m_{\xi^{1}}\left(x\right) = (x-\xi^{-q})(x-\xi^{-1})(x-\xi)(x-\xi^{q}) = (x+\xi)(x-\xi^{-1})(x-\xi)(x+\xi^{-1})=x^4-\left(\xi^2+\xi^{-2}\right)x^2+1$ and $m_{\xi^{0}}\left(x\right)m_{\xi^{2}}\left(x\right)\mid x^{q+1}-1$. Therefore, $g\left(x\right)\mid \big(x^4-\left(\xi^2+\xi^{-2}\right)x^2+1\big)\cdot \left(x^{q+1}-1\right)$, which implies that $\left[x^4-\left(\xi^2+\xi^{-2}\right)x^2+1\right]\cdot\left(x^{q+1}-1\right)\in \mathcal{C}$. Thus $d_H(\mathcal{C})\leq \Wt_H\big([x^4-(\xi^2+\xi^{-2})x^2+1]\cdot (x^{q+1}-1)\big)=6$. Hence $d_H(\mathcal{C})=6$.
        \end{itemize}
    \end{proof}
    Recall the decomposition of a cyclic code into constacyclic codes in Section \ref{sec:decomposition_of_cyclic_codes}. Notice that $m_{\xi^{0}}\left(x\right)m_{\xi^{2}}\left(x\right)\mid (x^{q+1}-1)$ and $m_{\xi{-1}}(x)m_{\xi^{1}}(x)\mid (x^{q+1}+1)$. We have the following decomposition of $\mathcal{C}$, $$\mathcal{C}=\mathcal{C}_1\curlyvee \mathcal{C}_2=\{(\mathbf{u}+\mathbf{v},\mathbf{u}-\mathbf{v}): \mathbf{u}\in \mathcal{C}_1, \mathbf{v} \in \mathcal{C}_{2} \},$$
    where $\mathcal{C}_1 = \langle m_{\xi^{0}}(x)m_{\xi^{2}}(x)\rangle \subseteq \mathbb{F}_q[x]/\langle x^{\frac{n}{2}}-1\rangle$ is a cyclic code and $\mathcal{C}_2 = \langle m_{\xi^{-1}}\left(x\right)m_{\xi^{1}}\left(x\right) \rangle \subseteq \mathbb{F}_q[x]/\langle x^{\frac{n}{2}}+1\rangle$ is a negacyclic code. Applying the orthogonal relation between the codewords of the negacyclic code and its dual code, we can exclude the existence of certain words in $\mathcal{C}$ and show that  $\mathcal{C}$ is a $(2q+2,9)_q$ MDS symbol-pair code. The following lemma establishes the generator polynomial of the dual code of the negacyclic code $\mathcal{C}_2$.
    \begin{lemma}
        \label{lemma14}
        Let $q > 3$ and $\mathcal{C}_2 = \langle m_{\xi^{-1}}\left(x\right)m_{\xi^{1}}\left(x\right) \rangle $ be the negacyclic code in $\mathbb{F}_q[x]/\langle x^{\frac{n}{2}} + 1\rangle $. Then, $\mathcal{C}_2^\bot$ has the generator polynomial $b(x) = b_0 + b_1x + b_2x^2 + \cdots+ b_{\frac{n}{2}-4}x^{\frac{n}{2}-4}$, where $b_{k} = \frac{\xi^{k+4}-\xi^{-k}}{\xi^4-1}$ for even $k$ and $b_{k} =0$ for odd $k$.

    \end{lemma}
    \begin{proof}
        Note that $\xi^q = -\xi^{-1}$, thus $m_{\xi^{-1}}\left(x\right)m_{\xi^{1}}\left(x\right)  = x^4-\beta x^2+1$, where $\beta = \xi^2+\xi^{-2}.$ Thus the reciprocal polynomial of $m_{\xi^{-1}}\left(x\right)m_{\xi^{1}}\left(x\right)$ is itself. Then, the generator polynomial of $\mathcal{C}^\bot$ is 
        $$b(x) = \frac{x^{q+1}+1}{x^4-\beta x^2+1 }.$$
        Express $b(x)$ of the form $$b(x) = b_0 + b_1x + b_2x^2 + \cdots + b_{q-3}x^{q-3},$$ where $b_k$ $(0\leq k\leq q-3)$ are undetermined coefficients over $\mathbb{F}_q$. Then 
        \begin{equation}
            \label{Eq::equation11}
                x^{q+1}+1 = (1-\beta x^2 + x^4)\cdot (b_0 + b_1x + b_2 x^{2} + \cdots + b_{q-3}x^{q-3}).
        \end{equation}
        Comparing the coefficients of two sides of (\ref{Eq::equation11}), we can obtain that $b_0 = 1, b_1 = 0,b_2=\beta, b_3=0$ and $b_k= \beta b_{k-2}-b_{k-4}$ for $k \geq 4$.
        The characteristic equation of the sequence $\{b_0, b_1, \ldots, b_{\frac{n}{2}-4}\}$ is $z^4 -\beta z^2 + 1 = 0,$ which has four distinct roots $z_1=\xi$, $z_2=\xi^{-1}$, $z_3=-\xi^{-1}$, and $z_4=-\xi$. 
        Then $b_k = B_1 \xi^k + B_2 (-\xi)^k + B_3 (-\xi^{-1})^k + B_4\xi^{-k},$ for some $B_1,B_2,B_3,B_4 \in \mathbb{F}_{q^2}$ (see \cite[Chapter 8]{Lidl1997}). From $b_0 = 1, b_1 = 0,b_2=\beta, b_3=0$, we can obtain that $B_1=B_2=-\frac{\xi^4}{2(1-\xi^4)}$ and $B_3=B_4=\frac{1}{2(1-\xi^4)}$. Hence, 
        \begin{equation*}
                b_k =-\frac{\xi^4}{2(1-\xi^4)}\xi^k - \frac{\xi^4}{2(1-\xi^4)}(-\xi)^k + \frac{1}{2(1-\xi^4)}(-\xi^{-1})^k + \frac{1}{2(1-\xi^4)}\xi^{-k}.
        \end{equation*}
      It can be directly verified that $b_k=0$ if $k$ is odd; and $b_k = \frac{\xi^{k+4}-\xi^{-k}}{\xi^{4}-1}$ if $k$ is even. The desired result follows.
    \end{proof}

    We now consider the minimum pair distance of $\mathcal{C}$. Firstly, we prove that $\mathcal{C}$ does not contain a codeword of the form $(\star, \star,\star,\star,\mathbf{0}_r,\star, \mathbf{0}_{n-r-6})$ for any $1 \leq r \leq n-7$.
    
    \begin{lemma}
        \label{lemma15}
        There exists no codeword of the form $\left(\star,\star,\star,\star,\star,\mathbf{0}_r,\star,\mathbf{0}_{n-r-6}\right)$ for any $1\leq r \leq n-7$.
    \end{lemma}
    \begin{proof}
        We prove the lemma by contradiction. Suppose that there exists a codeword $\mathbf{c} \in \mathcal{C}$ of the form $\left(\star,\star,\star,\star,\star,\mathbf{0}_r,\star,\mathbf{0}_{n-r-6}\right)$. Since $\gcd(2,q) = 1 $, w.l.o.g, we can assume that $\mathbf{c} = \left(2,2u_1,2u_2,2u_3,2u_4,\mathbf{0}_{\ell-5},2u_{\ell},\mathbf{0}_{n-r-6}\right)$, where $u_1,u_2,u_3,u_4,u_{\ell}\in\mathbb{F}_q^*$, and $l=r+5$.   By the $[\mathbf{u}+\mathbf{v}|\mathbf{u}-\mathbf{v}]$-decomposition, we suppose 
        \[\mathbf{c}=(\mathbf{u_c}+\mathbf{v_c}|\mathbf{u_c}-\mathbf{v_c}),\]
         where $\mathbf{u_c} \in \mathcal{C}_1$ and $\mathbf{v_c} \in  \mathcal{C}_2$. Next, we divide our discussion into three cases according to the structure of components of $\mathbf{c}$.
        \begin{itemize}
            \item [(\romannumeral1)] $6\leq \ell\leq \frac{n}{2}-1$, i.e., $\mathbf{c} = \left(2,2u_1,2u_2,2u_3,2u_4,\mathbf{0}_{\ell-5},2u_{\ell},\mathbf{0}_{\frac{n}{2}-1-\ell} ~| ~\mathbf{0}_{\frac{n}{2}}\right)$. 
            Then we can obtain $$\mathbf{u_c}=\mathbf{v_c} =\left(1,u_1,u_2,u_3,u_4,\mathbf{0}_{\ell-5},u_{\ell},\mathbf{0}_{\frac{n}{2}-1-\ell}\right). $$
            This means that $\left(1,u_1,u_2,u_3,u_4,\mathbf{0}_{\ell-5},u_{\ell},\mathbf{0}_{\frac{n}{2}-1-\ell}\right)\in \mathcal{C}_1 \cap \mathcal{C}_2$. Thus, we have $m_{\xi^{0}}(x)m_{\xi^{2}}(x)\mid \mathbf{u_c}(x)$ and $m_{\xi{-1}}(x)m_{\xi^{1}}(x)\mid \mathbf{u_c}(x)$. Then we obtain the following system:
            \begin{equation*}
                \begin{cases}
                    1+u_1+u_2+u_3+u_4+u_{\ell}&=0, \\
                    1+u_1\xi+u_2\xi^2+u_3\xi^3+u_4\xi^4+u_{\ell}\xi^{\ell}&=0,\\
                    1-u_1\xi+u_2\xi^2-u_3\xi^3+u_4\xi^4+u_{\ell}(-\xi)^{\ell}&=0,\\
                    1+u_1\xi^2+u_2\xi^4+u_3\xi^6+u_4\xi^8+u_{\ell}\xi^{2\ell}&=0.\\
                \end{cases}
            \end{equation*}
            
            If $l$ is even, we can get $u_1\xi+u_3\xi^3 =0$. Since $\xi^2 \not \in \mathbb{F}_q$, we obtain that $u_1 = u_3=0 $, which contradicts the assumption that $u_1,u_3\in\mathbb{F}_q^*$.
            
            If $l$ is odd, then we have the following system:
                \begin{equation}
                    \label{Eq::Lemma15-1.1}
                    \begin{cases}
                        1+u_1+u_2+u_3+u_4+u_{\ell}&=0, \\
                        1+u_2\xi^2+u_4\xi^4&=0,\\
                        u_1+u_3\xi^2+u_{\ell}\xi^{\ell-1}&=0,\\
                        1+u_1\xi^2+u_2\xi^4+u_3\xi^6+u_4\xi^8+u_{\ell}\xi^{2\ell}&=0.\\
                    \end{cases}
                \end{equation}
                Note that $f(x)= 1+u_2x^2+u_4x^4$ has two distinct roots $x=\xi$  and $x=-\xi=\xi^{-q}$. We can get that $m_{\xi{-1}}(x)m_{\xi^{1}}(x) \mid f(x)$. It is clear that $u_2 = -(\xi^2+\xi^{-2})$ and $u_4=1$. Substituting into the system (\ref{Eq::Lemma15-1.1}),
                we have the following system:
                \begin{equation*}
                    \begin{cases}
                        u_1+u_3+u_{\ell} &=\xi^2+\xi^{-2}-2 = (\xi^2-1)(1-\xi^{-2}),  \\
                        u_1+u_3\xi^2+u_{\ell}\xi^{\ell-1}&=0,\\
                        u_1\xi^2+u_3\xi^6+u_{\ell}\xi^{2\ell}&=\xi^6+\xi^2-1-\xi^8=(\xi^2-1)(1-\xi^6). \\
                    \end{cases}
                    \Longrightarrow 
                    \begin{cases}
                        u_1=  \frac{(1-\xi^{-2})(\xi^{\ell+1}-\xi^4-\xi^2-1)}{\xi^{\ell-1}-1},\\
                        u_3= \frac{\xi^{\ell+1}-\xi^{\ell-1}-\xi^6+1}{\xi^4-\xi^{\ell+1}},\\
                        u_{\ell} =  \frac{(\xi^{-2}-1)(\xi^2-1)(\xi^2+\xi^{-2})}{(\xi^{\ell-3}-1)(\xi^{\ell-1}-1)}.
                    \end{cases}
                \end{equation*}
                Note that $u_{\ell}\in \mathbb{F}_q^*$,
                \begin{equation*}
                    \begin{aligned}
                        u_{\ell}^q &= \frac{(\xi^{2}-1)(\xi^{-2}-1)(\xi^{-2}+\xi^{2})\xi^{2\ell-4}}{(1-\xi^{\ell-3})(1 -\xi^{\ell-1})} = \frac{(\xi^{-2}-1)(\xi^2-1)(\xi^2+\xi^{-2})}{(\xi^{\ell-3}-1)(\xi^{\ell-1}-1)}=u_{\ell}.
                    \end{aligned}
                \end{equation*}
            This means that $\xi^{2\ell-4} =1$, which is a contradiction since $\xi$ is a primitive $n$-th root of unity and $6\leq \ell\leq \frac{n}{2}-1$.

            \item [(\romannumeral2)] $\ell = \frac{n}{2}$, i.e., $\mathbf{c} = \left(2,2u_1,2u_2,2u_3,2u_4,\mathbf{0}_{\frac{n}{2}-5} ~|~   2u_{\ell},\mathbf{0}_{\frac{n}{2}-1} \right)$. Then
            $$\mathbf{u_c} = \left(1+u_{\ell},u_1,u_2,u_3,u_4,\mathbf{0}_{\frac{n}{2}-5}\right)$$
            and 
            $$\mathbf{v_c} = \left(1-u_{\ell},u_1,u_2,u_3,u_4,\mathbf{0}_{\frac{n}{2}-5}\right).$$
            Note that the generator polynomial of $\mathcal{C}_2$ is $x^4-(\xi^2+\xi^{-2})x^2+1$, thus $\xi$ and $-\xi$ are roots of $\mathbf{v_c}(x)$. Then we can easily obtain $u_1 =u_3=0$,  which contradicts the assumption that $u_1,u_3\in\mathbb{F}_q^*$.

            The discussion of the following cases is similar, so we omit it.
            \[\ell =\frac{n}{2}+1,\textnormal{ i.e.}, \mathbf{c} = \left(2,2u_1,2u_2,2u_3,2u_4,\mathbf{0}_{\frac{n}{2}-5}~|~ 0, 2u_{\ell},\mathbf{0}_{\frac{n}{2}-2} \right);\]
            \[\ell = \frac{n}{2}+2,\textnormal{ i.e.}, \mathbf{c} = \left(2,2u_1,2u_2,2u_3,2u_4,\mathbf{0}_{\frac{n}{2}-5}~|~0,0, 2u_{\ell},\mathbf{0}_{\frac{n}{2}-3} \right); \]
            \[\ell = \frac{n}{2}+3,\textnormal{ i.e.}, \mathbf{c} = \left(2,2u_1,2u_2,2u_3,2u_4,\mathbf{0}_{\frac{n}{2}-5}~|~ 0,0,0, 2u_{\ell},\mathbf{0}_{\frac{n}{2}-4}, \right);\]
            \[\ell = \frac{n}{2}+4,\textnormal{ i.e.}, \mathbf{c} = \left(2,2u_1,2u_2,2u_3,2u_4,\mathbf{0}_{\frac{n}{2}-5}~|~ 0,0,0,0, 2u_{\ell},\mathbf{0}_{\frac{n}{2}-5} \right). \]

            \item [(\romannumeral3)] $\frac{n}{2} +5 \leq \ell \leq n-2$, i.e., $\mathbf{c} = \left(2,2u_1,2u_2,2u_3,2u_4,\mathbf{0}_{\frac{n}{2}-5}~| ~\mathbf{0}_{\ell-\frac{n}{2}}, 2u_{\ell},\mathbf{0}_{n-\ell-1} \right)$. Denote $s=\ell-\frac{n}{2}$. Then 
            $$\mathbf{u_c} =\left(1,u_1,u_2,u_3,u_4,\mathbf{0}_{s-5},u'_{s},\mathbf{0}_{\frac{n}{2}-s-1}\right)$$
            and 
            $$\mathbf{v_c} =\left(1,u_1,u_2,u_3,u_4,\mathbf{0}_{s-5},-u'_{s},\mathbf{0}_{\frac{n}{2}-s-1}\right),$$  where $u'_s=u_{\ell}.$ Then we have the following system:
            \begin{equation*}
                \begin{cases}
                    1+u_1+u_2+u_3+u_4+u'_{s}&=0, \\
                    1+u_1\xi+u_2\xi^2+u_3\xi^3+u_4\xi^4-u'_{s}\xi^{s}&=0,\\
                    1-u_1\xi+u_2\xi^2-u_3\xi^3+u_4\xi^4-u'_{s}(-\xi)^{s}&=0,\\
                    1+u_1\xi^2+u_2\xi^4+u_3\xi^6+u_4\xi^8+u'_{s}\xi^{2s}&=0.\\
                \end{cases}
            \end{equation*}
            If $s$ is even, we can get $u_1\xi+u_3\xi^3 =0$. Since $\xi^2 \not \in \mathbb{F}_q$, we obtain $u_1 = u_3=0 $, which contradicts the assumption that $u_1,u_3\in\mathbb{F}_q^*$.
            
            If $s$ is odd, then we have the following system:
            \begin{equation}
                \label{Eq::Equation13}
                \begin{cases}
                    1+u_1+u_2+u_3+u_4+u'_{s}&=0, \\
                    1+u_2\xi^2+u_4\xi^4&=0,\\
                    u_1+u_3\xi^2-u'_{s}\xi^{s-1}&=0,\\
                    1+u_1\xi^2+u_2\xi^4+u_3\xi^6+u_4\xi^8+u'_{s}\xi^{2s}&=0.\\
                \end{cases}
            \end{equation}
            It is clear that $u_2 = -(\xi^2+\xi^{-2})$ and $u_4=1$. Putting the values into the system (\ref{Eq::Equation13}), we have the following system:
            \begin{equation*}
                \begin{cases}
                    u_1+u_3+u_{s}' &=(\xi^2-1)(1-\xi^{-2}),  \\
                    u_1+u_3\xi^2-u'_{s}\xi^{s-1}&=0,\\
                    u_1\xi^2+u_3\xi^6+u'_{s}\xi^{2s}&=(\xi^2-1)(1-\xi^6). \\
                \end{cases}
                \Longrightarrow 
                \begin{cases}
                    u_1=  \frac{(1-\xi^{-2})(\xi^{s+1}+\xi^4+\xi^2+1)}{\xi^{s-1}+1},\\
                    u_3= \frac{-\xi^{s+1}+\xi^{s-1}-\xi^6+1}{\xi^{s+1}-\xi^4},\\
                    u'_{s} = \frac{(\xi^{-2}-1)(\xi^2-1)(\xi^2+\xi^{-2})}{(\xi^{s-3}+1)(\xi^{s-1}+1)}.
                \end{cases}
            \end{equation*}
                Note that $u'_{s}=u_{\ell}\in \mathbb{F}_q^*$,
                \begin{equation*}
                    \begin{aligned}
                        (u_{s}')^q &= \frac{(\xi^{2}-1)(\xi^{-2}-1)(\xi^{-2}+\xi^{2})\xi^{2s-4}}{(1+\xi^{s-3})(1+\xi^{s-1})} = \frac{(\xi^{-2}-1)(\xi^2-1)(\xi^2+\xi^{-2})}{(\xi^{s-3}+1)(\xi^{s-1}+1)}=u_s'.
                    \end{aligned}
                \end{equation*}
                This means that $\xi^{2s-4} =1$, which is a contradiction for $5\leq s\leq \frac{n}{2}-2$. 
        \end{itemize}
        In summary, there does not exist such a codeword in $\mathcal{C}$.
    \end{proof}

    \begin{lemma}
        \label{lemma16}
        There exists no codeword of the form $\left(\star,\star,\star,\star,\mathbf{0}_r,\star,\star,\mathbf{0}_{n-r-6}\right)$ for any $1\leq r \leq n-7$.
    \end{lemma}
    \begin{proof}
        Similarly as in the proof of Lemma 15, we prove the lemma by contradiction and assume that $\mathbf{c} = \left(2,2u_1,2u_2,2u_3,\mathbf{0}_{\ell-4},2u_{\ell},2u_{\ell+1},\mathbf{0}_{n-\ell-2}\right)\in\mathcal{C}, $ where $u_1,u_2,u_3,u_{\ell},u_{\ell+1}\in \mathbb{F}^*_q$ and $\ell=r+4$. Further,
        $$\mathbf{c}=(\mathbf{u_c}+\mathbf{v_c}|\mathbf{u_c}-\mathbf{v_c}),$$
        where $\mathbf{u_c}\in \mathcal{C}_1$ and $\mathbf{v_c}\in \mathcal{C}_2$. Next, we divide our discussion into four cases according to the structure of components of $\mathbf{c}$.
        \begin{itemize}
            \item [(\romannumeral1)] $5\leq \ell\leq \frac{n}{2}-2$, i.e., $\mathbf{c} = \left(2,2u_1,2u_2,2u_3,\mathbf{0}_{\ell-4},2u_{\ell},2u_{\ell+1},\mathbf{0}_  {\frac{n}{2}-2-\ell} ~|~ \mathbf{0}_{\frac{n}{2}}\right)$. Then we can obtain $$\mathbf{u_c}=\mathbf{v_c} =\left(1,u_1,u_2,u_3,\mathbf{0}_{\ell-4},u_{\ell},u_{\ell+1},\mathbf{0}_{\frac{n}{2}-2-\ell}\right).$$
            This means that $\left(1,u_1,u_2,u_3,\mathbf{0}_{\ell-4},u_{\ell},u_{\ell+1},\mathbf{0}_{\frac{n}{2}-2-\ell}\right)\in \mathcal{C}_1 \cap \mathcal{C}_2$. 
            Thus, we have $m_{\xi^{0}}(x)m_{\xi^{2}}(x)\mid \mathbf{u_c}(x)$ and $m_{\xi{-1}}(x)m_{\xi^{1}}(x)\mid \mathbf{u_c}(x)$.
            Then we obtain the following system:
            \begin{equation*}
                \left\{
                \begin{aligned}
                    &1 + u_1 + u_2+u_3+u_\ell+u_{\ell+1} &= 0,\\
                    &1 + u_1\xi + u_2\xi^2+u_3\xi^3+u_\ell\xi^\ell+u_{\ell+1}\xi^{\ell+1} &= 0,\\
                    &1 - u_1\xi +  u_2 \xi^{2} - u_3 \xi^{3}+ u_\ell(-\xi)^{\ell}+ u_{\ell+1} (-\xi)^{\ell+1}& = 0,\\
                    &1 + u_1\xi^2 + u_2\xi^4+u_3\xi^6+u_\ell\xi^{2\ell}+u_{\ell+1}\xi^{2(\ell+1)} &= 0,\\
                \end{aligned}\right.
            \end{equation*}
                \textit{Case 1}: $\ell$ is odd. Then we have the following system:
                \begin{equation}
                    \label{Eq::Lemma16-1.1}
                    \begin{cases}
                        1 + u_1 + u_2+u_3+u_\ell+u_{\ell+1} &= 0,\\
                        1  + u_2\xi^2 + u_{\ell+1}\xi^{\ell+1} &= 0,\\
                        u_1  + u_3 \xi^{2} + u_\ell\xi^{\ell-1} & = 0,\\
                        1 + u_1\xi^2 + u_2\xi^4+u_3\xi^6+u_\ell\xi^{2\ell}+u_{\ell+1}\xi^{2(\ell+1)} &= 0.
                    \end{cases}
                \end{equation}
            Let $b(x) = b_0 + b_1x + b_2x^2 + \cdots+ b_{\frac{n}{2}-4}x^{\frac{n}{2}-4}$ be the generator polynomial of $\mathcal{C}_2^\bot$  given in Lemma \ref{lemma14}, then $\mathbf{e} = (-b_{\frac{n}{2}-4},0,0,0, b_0, b_1 ,\cdots, b_{\frac{n}{2}-5}) \in \mathcal{C}_2^\bot$. Thus
                \begin{equation}
                    \label{Eq::Lemma16-1.2}
                    \mathbf{e}\cdot \mathbf{v_c} = -b_{\frac{n}{2}-4} + u_\ell \cdot b_{\ell-4}+ u_{\ell+1} \cdot b_{\ell-3}=0.
                \end{equation}
                Note that $\ell$ is odd and $5\leq \ell \leq \frac{n}{2}-3$,
                we can get $b_{\frac{n}{2}-4} = 1$, $b_{\ell-4}=0$, and $b_{\ell-3}=\frac{\xi^{2\ell-2}-1}{\xi^{\ell-3}(\xi^4-1)}$ by Lemma \ref{lemma14}. From the second equation of (\ref{Eq::Lemma16-1.1}) and equation (\ref{Eq::Lemma16-1.2}), we obtain
                \begin{equation*}
                    \begin{cases} 
                        u_2 =  -\frac{\xi^{2\ell+2}-1}{\xi^{2}(\xi^{2\ell-2}-1)},\\
                        u_{\ell+1} = \frac{(\xi^4-1)\xi^{\ell-3}}{\xi^{2\ell-2}-1}.
                    \end{cases}
                \end{equation*}
                Substituting into the system (\ref{Eq::Lemma16-1.1}), we have the following system:
                \begin{equation*}
                    \begin{cases}
                        u_1 + u_3 + u_{\ell} & = \frac{(\xi^2-1)(\xi^{l+1}-1)(\xi^{l-1}-1)}{\xi^2(\xi^{2l-2}-1)},  \\
                        u_1  + \xi^{2} u_3  + u_\ell\xi^{\ell-1} & = 0,\\
                        u_1\xi^2  + u_3\xi^{6}   + u_\ell\xi^{2\ell} & = \frac{(\xi^{2\ell+2}+\xi^{2\ell}-\xi^{\ell-1}-1)(1-\xi^{\ell-1})(\xi^2-1)}{\xi^{2\ell-2}-1}.
                    \end{cases}
                \end{equation*}
                By solving the system, we can obtain 
                \begin{equation*}
                    \left\{
                    \begin{aligned}
                        u_1 &= -\frac{\xi^{\ell+1}-1}{\xi^2(\xi^{\ell-1}-1)},\\
                        u_3 &= \frac{(\xi^{\ell+1}-1)(\xi^{\ell+1}+1)}{\xi^4(\xi^{\ell-1}+1)(\xi^{\ell-3}-1)},\\
                        u_{\ell} &= \frac{(\xi^2-1)(\xi^{\ell+1}-1)(\xi^{\ell+1}+\xi^{\ell-1}-\xi^2+1)}{\xi^4(\xi^{2\ell-2}-1)(1-\xi^{\ell-3})}.\\
                    \end{aligned}\right.
                \end{equation*}
                Note that $u_3\in\mathbb{F}_q^*$,
                \begin{equation*}
                    u_{3}^q  
                    = \frac{(1-\xi^{\ell+1})(1+\xi^{\ell+1})}{\xi^{2}(1+\xi^{\ell-1})(1-\xi^{\ell-3})}
                    =\frac{(\xi^{\ell+1}-1)(\xi^{\ell+1}+1)}{\xi^4(\xi^{\ell-1}+1)(\xi^{\ell-3}-1)}=u_3.
                \end{equation*}
                We deduce that $\xi^2=1$, which is a contradiction since $\xi$ is a primitive $n$-th root of unity.

                \textit{Case 2}: $\ell$ is even. Then we have the following system:
                \begin{equation}
                    \label{Eq::Lemma16-1.3}
                    \left\{
                    \begin{aligned}
                        &1 + u_1 + u_2+u_3+u_\ell+u_{\ell+1} &= 0,\\
                        &1  + u_2\xi^2 + u_\ell\xi^{\ell}  &= 0,\\ 
                        &u_1  + u_3 \xi^{2} + u_{\ell+1}\xi^{\ell}& = 0,\\
                        &1 + u_1\xi^2 + u_2\xi^4+u_3\xi^6+u_\ell\xi^{2\ell}+u_{\ell+1}\xi^{2(\ell+1)} &= 0.
                    \end{aligned}\right.
                \end{equation}
                From the second and third equations of the system, we obtain $(u_3-u_1u_2) \xi^{2} + (u_{\ell+1}-u_1u_{\ell})\xi^{\ell} =0 \Longrightarrow (u_3-u_1u_2) + (u_{\ell+1}-u_1u_{\ell})\xi^{\ell-2} =0$. Note that $\xi^{\ell-2} \not \in \mathbb{F}_q$. Therefore, we obtain $u_3=u_1u_2$ and $u_{\ell+1}=u_1u_{\ell}$. Then we can derive $$1 + u_1 + u_2+u_3+u_\ell+u_{\ell+1} = (1+u_1)(1+u_2+u_{\ell}) = 0.$$  
            Thus $u_1 = -1$ or $1+u_2+u_{\ell}=0$.
            
            If $u_1 = -1$, then $u_3=-u_2$ and $u_{\ell+1}=-u_{\ell}$. Then 
                    \begin{equation*}
                            1 + u_1\xi^2 + u_2\xi^4+u_3\xi^6+u_\ell\xi^{2\ell}+u_{\ell+1}\xi^{2(\ell+1)}
                            =(1-\xi^2)(1+u_2\xi^4+u_{\ell}\xi^{2\ell})
                            =0.
                    \end{equation*}
                    Note that $1-\xi^2 \neq 0$, thus $1+u_2\xi^4+u_{\ell}\xi^{2\ell}=0$. Combining it with the second equation of (\ref{Eq::Lemma16-1.3}), we have the following system:
                    \begin{equation*}
                        \begin{cases}
                            1  + u_2\xi^2 + u_\ell\xi^{\ell}  &= 0,\\ 
                            1+u_2\xi^4+u_{\ell}\xi^{2\ell}&=0.
                        \end{cases}
                        \Longrightarrow
                        \begin{cases}
                            u_{2} = \frac{1-\xi^{\ell}}{\xi^{\ell+2}-\xi^4},\\
                            u_{\ell} = \frac{\xi^2-1}{\xi^{2\ell}-\xi^{\ell+2}}.
                        \end{cases}
                    \end{equation*}
                    Note that $u_2 \in\mathbb{F}_q^*$, 
                    $$u_2^q = \frac{\xi^{\ell+2}-\xi^{2}}{1-\xi^{\ell-2}} = u_2 = \frac{1-\xi^{\ell}}{\xi^{\ell+2}-\xi^4}.$$
                    This means $\xi^6=1$, which is a contradiction since $\xi$ is a primitive $n$-th root of unity.
         
                    If $1+u_2+u_{\ell}=0$, then combining it with the second equation of (\ref{Eq::Lemma16-1.3}), we can obtain
                    \begin{equation*}
                        \begin{cases}
                            1  + u_2\xi^2 + u_\ell\xi^{\ell}  &= 0,\\ 
                            1+u_2+u_{\ell}&=0.
                        \end{cases}
                        \Longrightarrow
                        \begin{cases}
                            u_{2} = \frac{1-\xi^{\ell}}{\xi^{\ell}-\xi^2},\\
                            u_{\ell} = \frac{\xi^2-1}{\xi^{\ell}-\xi^{2}}.
                        \end{cases}
                    \end{equation*}
                    Note that $u_2 \in\mathbb{F}_q^*$,  
                    $$u_2^q = \frac{1-\xi^{-\ell}}{\xi^{-\ell}-\xi^{-2}}=\frac{\xi^{\ell}-1}{1-\xi^{\ell-2}} = u_2 = \frac{1-\xi^{\ell}}{\xi^{\ell}-\xi^2}.$$
                    This means that $\xi^2=1$, which is a contradiction.

            \item [(\romannumeral2)] $\ell = \frac{n}{2}-1$, i.e., $\mathbf{c} = \left(2,2u_1,2u_2,2u_3,\mathbf{0}_{\frac{n}{2}-5}, 2u_{\ell} ~|~  2u_{\ell+1},\mathbf{0}_{\frac{n}{2}-1} \right).$ Then
            $$\mathbf{u_c} = \left(1+u_{\ell+1},u_1,u_2,u_3,\mathbf{0}_{\frac{n}{2}-5},u_\ell\right)$$ 
            and $$\mathbf{v_c} = \left(1-u_{\ell+1},u_1,u_2,u_3,\mathbf{0}_{\frac{n}{2}-5},u_\ell\right).$$
            Since $\mathcal{C}_2$ is a negacyclic code, the negacyclic shift of $\mathbf{v_c}$, given by $$\mathbf{w_c} = \left(-u_\ell,1-u_{\ell+1},u_1,u_2,u_3,\mathbf{0}_{\frac{n}{2}-5}\right),$$ belongs to $\mathcal{C}_2$. Note that the generator polynomial of $\mathcal{C}_2$ is $x^4-(\xi^2+\xi^{-2})x^2+1$, thus $\xi$ and $-\xi$ are roots of $\mathbf{w_c}(x)$. Then we can easily obtain $u_2=0$, which contradicts the assumption that $u_2\in\mathbb{F}_q^*$.

            The case of $\ell = \frac{n}{2}+3$, i.e. $\mathbf{c}=\left(2,2u_1,2u_2,2u_3,\mathbf{0}_{\frac{n}{2}-4} ~|~ 0,0,0,2u_{\ell}, 2u_{\ell+1},\mathbf{0}_{\frac{n}{2}-5} \right)$ is similar, so we omit it.

            \item [(\romannumeral3)] $\ell = \frac{n}{2}$, i.e., $\mathbf{c} = \left(2,2u_1,2u_2,2u_3,\mathbf{0}_{\frac{n}{2}-4} ~|~ 2u_{\ell} , 2u_{\ell+1},\mathbf{0}_{\frac{n}{2}-2} \right)$. Then
            $$\mathbf{u_c} = \left(1+u_{\ell},u_1+u_{\ell+1},u_2,u_3,\mathbf{0}_{\frac{n}{2}-4}\right)$$
            and 
            $$\mathbf{v_c} = \left(1-u_{\ell},u_1-u_{\ell+1},u_2,u_3,\mathbf{0}_{\frac{n}{2}-4}\right).$$
            Note that the generator polynomial of $\mathcal{C}_2$ has degree 4, while $\mathbf{v_c}(x)$ is a codeword polynomial of degree 3, which is a contradiction.

            The discussion of the following cases is similar, so we omit it.

            \[\ell = \frac{n}{2}+1, \textnormal{ i.e.},\mathbf{c}=\left(2,2u_1,2u_2,2u_3,\mathbf{0}_{\frac{n}{2}-4} ~|~ 0,2u_{\ell}, 2u_{\ell+1},\mathbf{0}_{\frac{n}{2}-3} \right);\]
            \[\ell = \frac{n}{2}+2, \textnormal{ i.e.},\mathbf{c}=\left(2,2u_1,2u_2,2u_3,\mathbf{0}_{\frac{n}{2}-4} ~|~ 0,0,2u_{\ell}, 2u_{\ell+1},\mathbf{0}_{\frac{n}{2}-4} \right);\]
            
            \item [(\romannumeral4)] $\frac{n}{2}+4\leq \ell\leq n-3$, i.e., $\mathbf{c} = \left(2,2u_1,2u_2,2u_3,\mathbf{0}_{\frac{n}{2}-4} ~|~\mathbf{0}_{\ell-\frac{n}{2}},2u_{\ell},2u_{\ell+1},\mathbf{0}_{n-2-\ell}\right)$. Denote $s = \ell -\frac{n}{2}$.
            Then
            $$\mathbf{u_c} =\left(1,u_1,u_2,u_3,\mathbf{0}_{s-4},u_s',u_{s+1}',\mathbf{0}_{\frac{n}{2}-s-2}\right)$$
            and 
            $$\mathbf{v_c} = \left(1,u_1,u_2,u_3,\mathbf{0}_{w-4},-u_s',-u_{s+1}',\mathbf{0}_{\frac{n}{2}-s-2}\right)$$ where $u_{s}' = u_{\ell}$ and $u_{s+1}' = u_{\ell+1}$.
            Then we have the following system:
            \begin{equation*}
                \begin{cases}
                    1 + u_1 + u_2+u_3+u_s'+u_{s+1}' &= 0,\\
                    1 + u_1\xi + u_2\xi^2+u_3\xi^3- u_s'\xi^{s}-u_{s+1}'\xi^{s+1} &= 0,\\
                    1 - u_1\xi +  u_2 \xi^{2} - u_3 \xi^{3}- u_s'(-\xi)^{s}- u_{s+1}' (-\xi)^{s+1}& = 0,\\
                    1 + u_1\xi^2 + u_2\xi^4+u_3\xi^6+u_s'\xi^{2s}+u_{s+1}'\xi^{2(s+1)} &= 0.
                \end{cases}
            \end{equation*}
            
            \textit{Case} 1: $s$ is odd. Then we have the following system:
            \begin{equation}
                \label{Eq::Lemma16-2.1}
                \begin{cases}
                    1 + u_1 + u_2+u_3+u_s'+u_{s+1}' &= 0,\\
                    1 + u_2\xi^2 -u_{s+1}'\xi^{s+1} &= 0,\\
                    u_1 +u_3\xi^2- u_s'\xi^{s-1}&= 0,\\
                    1 + u_1\xi^2 + u_2\xi^4+u_3\xi^6+u_s'\xi^{2s}+u_{s+1}'\xi^{2(s+1)} &= 0.
                \end{cases}
            \end{equation}
            Since $\mathbf{e} = (-b_{\frac{n}{2}-4},0,0,0,b_0,b_1,\cdots,b_{\frac{n}{2}-5})\in \mathcal{C}_2^{\bot}$, 
            \begin{equation}
                \label{Eq::Lemma16-2.2}
                \mathbf{e}\cdot \mathbf{v_c} = -b_{\frac{n}{2}-4} - u_s' \cdot b_{s-4}- u_{s+1}' \cdot b_{s-3} =0.
            \end{equation}
            Note that $s$ is odd and $5\leq s \leq \frac{n}{2}-3$, we can get $b_{s-4} = 0$, $b_{s-3} = \frac{\xi^{s+1}-\xi^{-(s-3)}}{\xi^4-1}$ and $b_{\frac{n}{2}-4} = 1$ by Lemma \ref{lemma14}. From the second equation
            of (\ref{Eq::Lemma16-2.1}) and equation (\ref{Eq::Lemma16-2.2}), we obtain
            \begin{equation*}
                \begin{cases}
                    u_2 &= -\frac{\xi^{2s+2}-1}{\xi^{2}(\xi^{2s-2}-1)},\\
                    u_{s+1}' &= \frac{(1-\xi^4)\xi^{s-3}}{\xi^{2s-2}-1}.
                \end{cases}
            \end{equation*}
            Putting the values into (\ref{Eq::Lemma16-2.1}), we have the following system:
            \begin{equation*}
                \begin{cases}
                    u_1 + u_3 + u_{s}' & =  \frac{(\xi^2-1)(\xi^{s+1}+1)(\xi^{s-1}+1)}{\xi^2(\xi^{2s-2}-1)},  \\
                    u_1  + \xi^{2} u_3 - u_s'\xi^{s-1} & = 0,\\
                    u_1\xi^2  + u_3\xi^{6} + u_s'\xi^{2s} &= \frac{(\xi^{2s+2}+\xi^{2s}+\xi^{s-1}-1)(1+\xi^{s-1})(\xi^2-1)}{\xi^{2s-2}-1}.
                \end{cases}
            \end{equation*}
            By solving the system, we can get 
            \begin{equation*}
                \left\{
                \begin{aligned}
                    u_1 &= -\frac{\xi^{s+1}+1}{\xi^2(\xi^{s-1}+1)},\\
                    u_3 &= \frac{(\xi^{s+1}+1)(\xi^{s+1}-1)}{\xi^4(\xi^{s-1}-1)(\xi^{s-3}+1)},\\
                    u_{s}' &= \frac{(\xi^2-1)(\xi^{s+1}+1)(\xi^{s+1}+\xi^{s-1}+\xi^2-1)}{\xi^4(\xi^{2s-2}-1)(1+\xi^{s-3})}.
                \end{aligned}\right.
            \end{equation*}
            Note that $u_3\in\mathbb{F}_q^*$,
            \begin{equation*}
                    u_{3}^q  
                    = \frac{(1+\xi^{s+1})(\xi^{s+1}-1)}{\xi^{2}(\xi^{s-1}-1)(1+\xi^{s-3})}
                    = \frac{(\xi^{s+1}+1)(1-\xi^{s+1})}{\xi^4(1-\xi^{s-1})(\xi^{s-3}+1)}=u_3.
            \end{equation*}
            This means that $\xi^2=1$, which is a contradiction since $\xi$ is a primitive $n$-th root of unity.
            
            \textit{Case} 2: $s$ is even. Then we have the following system:
                \begin{equation}
                    \label{Eq::Lemma16-2.3}
                    \left\{
                    \begin{aligned}
                        &1 + u_1 + u_2+u_3+u_s'+u_{s+1}' &= 0,\\
                        &1  + u_2\xi^2 - u_s'\xi^{s}  &= 0,\\ 
                        &u_1  + u_3 \xi^{2} - u_{s+1}'\xi^{s}& = 0,\\
                        &1 + u_1\xi^2 + u_2\xi^4+u_3\xi^6+u_s'\xi^{2s}+u_{s+1}'\xi^{2s+2} &= 0.
                    \end{aligned}\right.
                \end{equation}
                From the second and third equations of the system, we obtain $(u_3-u_1u_2) \xi^{2} - (u_{s+1}'-u_1u_{s}')\xi^{s} =0  \Longrightarrow (u_3-u_1u_2) + (u_{s+1}'-u_1u_{s}')\xi^{s-2} =0$. Note that $\xi^{s-2} \not \in \mathbb{F}_q$. Therefore, we have $u_3=u_1u_2$ and $u_{s+1}'=u_1u_{s}'$. Then we can derive $$1 + u_1 + u_2+u_3+u_s'+u_{s+1}' = (1+u_1)(1+u_2+u_{s}') = 0.$$  
                Thus $u_1 = -1$ or $1+u_2+u_{s}'= 0$. 

                 If $u_1 = -1$, then $u_3=-u_2$ and $u_{s+1}'=-u_{s}'$. Then 
                    \begin{equation*}
                            1 + u_1\xi^2 + u_2\xi^4+u_3\xi^6+u_s'\xi^{2s}+u_{s+1}'\xi^{2s+2}  
                            =(1-\xi^2)(1+u_2\xi^4+u_{s}'\xi^{2s})
                            =0.
                    \end{equation*}
                    Since $1-\xi^2 \neq 0$, then  $1+u_2\xi^4+u_{s}'\xi^{2s}=0$. Combining it with the second equation of (\ref{Eq::Lemma16-2.3}), we have the following system:
                    \begin{equation*}
                        \begin{cases}
                            1  + u_2\xi^2 - u_s'\xi^{s}  &= 0,\\ 
                            1+u_2\xi^4+u_{s}'\xi^{2s}&=0.
                        \end{cases}
                        \Longrightarrow
                        \begin{cases}
                            u_{2} = -\frac{1+\xi^{s}}{\xi^{s+2}+\xi^4},\\
                            u_{s}' = \frac{\xi^2-1}{\xi^{2s}+\xi^{s+2}}.
                        \end{cases}
                    \end{equation*}
                    Note that $u_2 \in\mathbb{F}_q^*$, then 
                    $$u_2^q = -\frac{\xi^{s+2}+\xi^{2}}{1+\xi^{s-2}} = -\frac{1+\xi^{s}}{\xi^{s+2}+\xi^4}= u_2.$$
                    This means that $\xi^6=1$, which is a contradiction since $\xi$ is a primitive $n$-th root of unity.

     If $1+u_2+u_{s}'=0$, then combining it with the second equation of (\ref{Eq::Lemma16-2.3}), we can obtain
                    \begin{equation*}
                        \begin{cases}
                            1+u_2+u_{s}'&=0,\\
                            1 + u_2\xi^2 - u_s'\xi^{s} &= 0. 
                        \end{cases}
                        \Longrightarrow
                        \begin{cases}
                            u_{2} = -\frac{1+\xi^{s}}{\xi^{s}+\xi^2},\\
                            u_{s}' = \frac{1-\xi^2}{\xi^{s}+\xi^2}.
                        \end{cases}
                    \end{equation*}
                    Note that $u_2 \in\mathbb{F}_q^*$, then 
                    $$u_2^q =  -\frac{1+\xi^{-s}}{\xi^{-s}+\xi^{-2}}=-\frac{\xi^{s}+1}{1+\xi^{s-2}} = u_2 = -\frac{1+\xi^{s}}{\xi^{s}+\xi^2}.$$
                    This means $\xi^2=1$, which is a contradiction since $\xi$ is a primitive $n$-th root of unity.
            \end{itemize}
        In summary, there does not exist such a codeword in $\mathcal{C}$.
    \end{proof}

    \begin{lemma}
        \label{lemma17}
        There exists no codeword of the form $\left(\star,\star,\star,\mathbf{0}_r,\star,\star,\star,\mathbf{0}_{n-r-6}\right)$ for any $1\leq r \leq n-7$.
    \end{lemma}

    \begin{proof}
        Similarly as in the proof of Lemmas 15 and 16, we prove the lemma by contradiction and assume that $\mathbf{c} = \left(2,2u_1,2u_2,\mathbf{0}_{\ell-3},2u_{\ell},2u_{\ell+1},u_{\ell+2},\mathbf{0}_{n-\ell-3}\right)=(\mathbf{u_c}+\mathbf{v_c}|\mathbf{u_c}-\mathbf{v_c})$, where $\ell=r+3$, $u_1,u_2,u_{\ell},u_{\ell+1},u_{\ell+2}\in \mathbb{F}^*_q$,  $\mathbf{u_c}\in \mathcal{C}_1$ and $\mathbf{v_c}\in \mathcal{C}_2$. Next, we divide our discussion into four cases according to the structure of components of $\mathbf{c}$.
        \begin{itemize}
            \item [(\romannumeral1)] $4\leq \ell\leq \frac{n}{2}-3$, i.e., $\mathbf{c} = \left(2,2u_1,2u_2,\mathbf{0}_{\ell-3},2u_{\ell},2u_{\ell+1},2u_{\ell+2},\mathbf{0}_  {\frac{n}{2}-3-\ell} ~|~ \mathbf{0}_{\frac{n}{2}}\right)$. Then we can obtain $$\mathbf{u_c}=\mathbf{v_c} =\left(1,u_1,u_2,\mathbf{0}_{\ell-4},u_{\ell},u_{\ell+1},u_{\ell+2},\mathbf{0}_{\frac{n}{2}-3-\ell}\right).$$ This means that $\left(1,u_1,u_2,\mathbf{0}_{\ell-3},u_{\ell},u_{\ell+1},u_{\ell+2},\mathbf{0}_{\frac{n}{2}-3-\ell}\right)\in \mathcal{C}_1 \cap \mathcal{C}_2$. Thus, we have $m_{\xi^{0}}(x)m_{\xi^{2}}(x)\mid \mathbf{u_c}(x)$ and $m_{\xi{-1}}(x)m_{\xi^{1}}(x)\mid \mathbf{v_c}(x)$. Then we obtain the following system:
            \begin{equation*}
                \begin{cases}
                    1 + u_1 + u_2+u_\ell+u_{\ell+1}+u_{\ell+2} &= 0,\\
                    1 + \xi u_1 + \xi^2u_2 +\xi^\ell u_\ell+\xi^{\ell+1} u_{\ell+1}+\xi^{\ell+2} u_{\ell+2} &= 0,\\
                    1 - \xi u_1 + \xi^2u_2+(-\xi)^\ell u_\ell+(-\xi)^{\ell+1}u_{\ell+1}+(-\xi)^{\ell+2} u_{\ell+2} &= 0,\\
                    1 + \xi^2 u_1 + \xi^4u_2+\xi^{2\ell}u_\ell+\xi^{2(\ell+1)}u_{\ell+1}+\xi^{2(\ell+2)}u_{\ell+2} &= 0.
                \end{cases}
            \end{equation*}
            If $\ell$ is even, we can get $ u_1  + \xi^{\ell} u_{\ell+1} = 0$. This implies that $\xi^{\ell}=-\frac{u_1}{u_{\ell+1}}\in \mathbb{F}_q^*$, which is a contradiction since it is easy to check that $\xi^{\ell}\notin\mathbb{F}_q^*$.
            
            If $\ell$ is odd, then we have the following system:
            \begin{equation}
                \label{Eq:Lemma17-1.1}
                \begin{cases}
                    1 + u_1 + u_2+u_\ell+u_{\ell+1}+u_{\ell+2} &= 0,\\
                    1  + \xi^2u_2 +\xi^{\ell+1} u_{\ell+1}&= 0,\\
                    u_1 + \xi^{\ell-1} u_\ell + \xi^{\ell+1} u_{\ell+2} &= 0,\\
                    1 + \xi^2 u_1 + \xi^4u_2+\xi^{2\ell}u_\ell+\xi^{2(\ell+1)}u_{\ell+1}+\xi^{2(\ell+2)}u_{\ell+2} &= 0,\\
                \end{cases}
            \end{equation}
            Since $\mathbf{e}=(-b_{\frac{n}{2}-4},0,0,0,b_0,b_1,\cdots,b_{\frac{n}{2}-5})\in \mathcal{C}_2^\bot$, 
            \begin{equation}
                \label{Eq:Lemma17-1.2}
                \mathbf{e}\cdot \mathbf{v_c} =-b_{\frac{n}{2}-4} + u_\ell \cdot b_{\ell-4}+ u_{\ell+1} \cdot b_{\ell-3}+ u_{\ell+2} \cdot b_{\ell-2}=0.
            \end{equation} Since $4\leq \ell \leq \frac{n}{2}-3$ and $\ell$ is odd, we can get $b_{\ell-4} = b_{\ell-2} = 0$, $b_{\ell-3} = \frac{\xi^{\ell+1}-\xi^{-(\ell-3)}}{\xi^4-1}$, and $b_{\frac{n}{2}-4}=1$ by Lemma \ref{lemma14}. 
            From the second equation of (\ref{Eq:Lemma17-1.1}) and equation (\ref{Eq:Lemma17-1.2}), we obtain
            \begin{equation*}
                \begin{cases}
                    u_2= -\frac{\xi^{2\ell+2}-1}{\xi^{2}(\xi^{2\ell-2}-1)}, \\
                    u_{\ell+1}=\frac{1}{b_{\ell-3}} = \frac{(\xi^4-1)\xi^{\ell-3}}{\xi^{2\ell-2}-1}.
                \end{cases}
            \end{equation*}
            Substituting into the system (\ref{Eq:Lemma17-1.1}), we have the following system:
            \begin{equation*}
                \begin{cases}
                    u_1  + u_{\ell} + u_{\ell+2}& = \frac{(\xi^2-1)(\xi^{l+1}-1)(\xi^{l-1}-1)}{\xi^2(\xi^{2l-2}-1)},  \\
                    u_1 + \xi^{\ell-1} u_\ell + \xi^{\ell+1} u_{\ell+2} &= 0,\\
                    u_1\xi^2  + u_{\ell}\xi^{2\ell}   + u_{\ell+2}\xi^{2\ell+4} & = \frac{(\xi^{2\ell+2}+\xi^{2\ell}-\xi^{\ell-1}-1)(1-\xi^{\ell-1})(\xi^2-1)}{\xi^{2\ell-2}-1}.
                \end{cases}
            \end{equation*}
            By solving the system, we can obatin 
            \begin{equation*}
                \left\{
                \begin{aligned}
                    u_1 &= \frac{(\xi^2-1)(\xi^{2\ell}-\xi^{\ell+1}-\xi^{\ell-1}-1)}{\xi^2(\xi^{2\ell-2}-1)},\\
                    u_{\ell} &= \frac{\xi^{\ell+1}-1}{\xi^{\ell+1}(\xi^{\ell-1}-1)},\\
                    u_{\ell+2} &= -\frac{\xi^{\ell+1}+1}{\xi^{\ell+1}(\xi^{\ell-1}+1)}.
                \end{aligned}\right.
            \end{equation*}
            Note that $u_{\ell}\in\mathbb{F}_q^*$,
            \begin{equation*}
                    u_{\ell}^q 
                    = \frac{(1-\xi^{\ell+1})\xi^{\ell-1}}{1-\xi^{\ell-1}}
                    =\frac{\xi^{\ell+1}-1}{\xi^{\ell+1}(\xi^{\ell-1}-1)}=u_{\ell}.
            \end{equation*}
            This means that $\xi^{2\ell}=1$, which is a contradiction since $\xi$ is a primitive $n$-th root of unity and $4 \leq \ell  \leq \frac{n}{2}-3$. 

            \item [(\romannumeral2)]  $\ell = \frac{n}{2}-2$, i.e., $\mathbf{c} = \left(2,2u_1,2u_2,\mathbf{0}_{\frac{n}{2}-5}, 2u_{\ell},2u_{\ell+1} ~|~ 2u_{\ell+2},\mathbf{0}_{\frac{n}{2}-1} \right)$. Then
            $$\mathbf{u_c} = \left(1+u_{\ell+2},u_1,u_2,\mathbf{0}_{\frac{n}{2}-5},u_\ell,u_{\ell+1}\right)$$ 
            and
            $$\mathbf{v_c} = \left(1-u_{\ell+2},u_1,u_2,\mathbf{0}_{\frac{n}{2}-5},u_\ell,u_{\ell+1}\right).$$
            Since $\mathcal{C}_2$ is a negacyclic code, the negacyclic shift of $\mathbf{v_c}$, given by $$\mathbf{w_c} = \left(-u_{\ell},-u_{\ell+1},1-u_{\ell+2},u_1,u_2,\mathbf{0}_{\frac{n}{2}-5}\right),$$ belongs to $\mathcal{C}_2$. Note that the generator polynomial of $\mathcal{C}_2$ is $x^4-(\xi^2+\xi^{-2})x^2+1$, thus $\xi$ and $-\xi$ are roots of $\mathbf{w_c}(x)$. Then we can easily obtain $u_1=u_{\ell+1}=0$, which contradicts the assumption that $u_1=u_{\ell+1}\in\mathbb{F}_q^*$.
            
            The case of $\ell = \frac{n}{2}+2$, i.e., $\mathbf{c} = \left(2,2u_1,2u_2,\mathbf{0}_{\frac{n}{2}-3} ~|~ 0,0,2u_{\ell}, 2u_{\ell+1},2u_{\ell+2},\mathbf{0}_{\frac{n}{2}-5} \right)$ is similar, so we omit it.

            \item [(\romannumeral3)] $\ell = \frac{n}{2}-1$, i.e., $\mathbf{c} = \left(2,2u_1,2u_2,\mathbf{0}_{\frac{n}{2}-4}, 2u_{\ell} ~|~ 2u_{\ell+1},2u_{\ell+2},\mathbf{0}_{\frac{n}{2}-2} \right)\in\mathcal{C}$. Then
            $$\mathbf{u_c} = \left(1+u_{\ell+1},u_1+u_{\ell+2},u_2,\mathbf{0}_{\frac{n}{2}-4},u_\ell\right)$$ 
            and $$\mathbf{v_c} = \left(1-u_{\ell+1},u_1-u_{\ell+2},u_2,\mathbf{0}_{\frac{n}{2}-4},u_\ell\right).$$
            Since $\mathcal{C}_2$ is a negacyclic code, the negacyclic shift $\mathbf{v_c}$, given by $$\mathbf{w_c}^\prime = \left(-u_\ell,1-u_{\ell+1},u_1-u_{\ell+2},u_2,\mathbf{0}_{\frac{n}{2}-4}\right),$$ belongs to $\mathcal{C}_2$. Note that the generator polynomial of $\mathcal{C}_2$ is $x^4-(\xi^2+\xi^{-2})x^2+1$, while the degree of $\mathbf{w_c}^{\prime}(x)$ is 3, which is a contradiction.
                
            The discussion of the following cases is similar, so we omit it.
            \begin{gather*}
                \ell = \frac{n}{2},\textnormal{ i.e.}, \mathbf{c} = \left(2,2u_1,2u_2,\mathbf{0}_{\frac{n}{2}-3} ~|~ 2u_{\ell}, 2u_{\ell+1},2u_{\ell+2},\mathbf{0}_{\frac{n}{2}-3} \right);\\
                \ell = \frac{n}{2}+1,\textnormal{ i.e.}, \mathbf{c} = \left(2,2u_1,2u_2,\mathbf{0}_{\frac{n}{2}-3} ~|~ 0,2u_{\ell}, 2u_{\ell+1},2u_{\ell+2},\mathbf{0}_{\frac{n}{2}-4} \right).
            \end{gather*}
            \item [(\romannumeral4)] $\frac{n}{2}+3\leq \ell \leq n-4$, i.e., $\mathbf{c} = \left(2,2u_1,2u_2,\mathbf{0}_{\frac{n}{3}-3} ~|~ \mathbf{0}_{\ell-\frac{n}{2}}, 2u_{\ell},2u_{\ell+1},2u_{\ell+2},\mathbf{0}_ {n-3-\ell} \right)$. 
            Denote $s={\ell-\frac{n}{2}}$. Then
            $$\mathbf{u_c} =\left(1,u_1,u_2,\mathbf{0}_{s-3},u_{s}^{\prime},u_{s+1}^{\prime},u_{s+2}^{\prime},\mathbf{0}_{\frac{n}{2}-s-3}\right)$$
            and 
            $$\mathbf{v_c} =\left(1,u_1,u_2,\mathbf{0}_{s-3},-u_{s}^{\prime},-u_{s+1}^{\prime},-u_{s+2}^{\prime},\mathbf{0}_{\frac{n}{2}-s-1}\right),$$  where $u_s^{\prime}=u_{\ell}$, $u_{s+1}^{\prime}=u_{\ell+1}$, and $u_{s+2}^{\prime}=u_{\ell+2}$. 
            Then we have the following system:
            \begin{equation*}
                \left\{
                \begin{aligned}
                    &1 + u_1 + u_2+u_s^{\prime}+u_{s+1}^{\prime}+u_{s+2}^{\prime} &= 0,\\
                    &1 + \xi u_1 + \xi^2u_2 -\xi^s u_s^{\prime}-\xi^{s+1} u_{s+1}^{\prime}-\xi^{s+2} u_{s+2}^{\prime}&= 0,\\
                    &1 - \xi u_1 + \xi^2u_2-(-\xi)^s u_s^{\prime}-(-\xi)^{s+1}u_{s+1}^{\prime}-(-\xi)^{s+2} u_{s+2}^{\prime} &= 0,\\
                    &1 + \xi^2 u_1 + \xi^4u_2+\xi^{2s}u_s^{\prime}+\xi^{2(s+1)}u_{s+1}^{\prime}+\xi^{2(s+2)}u_{s+2}^{\prime} &= 0.
                \end{aligned}\right.
            \end{equation*}
            If $s$ is even, we can get $u_1 - \xi^{s} u_{s+1}'= 0.$ This implies that $\xi^{s}=\frac{u_1}{u_{s+1}'}\in \mathbb{F}_q^*$, which a contradiction since it is easy to check that $\xi^{s}\notin\mathbb{F}_q^*$.
            
            If $s$ is odd, then we have the following system:
            \begin{equation}
                \label{Eq::lemma17-2.1}
                \begin{cases}
                    1 + u_1 + u_2+u_s^{\prime}+u_{s+1}^{\prime}+u_{s+2}^{\prime} &= 0,\\
                    1 + \xi^2u_2 -\xi^{s+1} u_{s+1}^{\prime}&= 0,\\
                    u_1 - \xi^{s-1} u_s^{\prime} - \xi^{s+1} u_{s+2}^{\prime} &= 0,\\
                    1 + \xi^2 u_1 + \xi^4u_2+\xi^{2s}u_s^{\prime}+\xi^{2(s+1)}u_{s+1}^{\prime}+\xi^{2(s+2)}u_{s+2}^{\prime} &= 0.
                \end{cases}
            \end{equation}
            Since $\mathbf{e}'=\left\{0,0,b_0,b_1,\cdots,b_{\frac{n}{2}-4},0\right\}\in \mathcal{C}_2^\bot$, 
            \begin{equation}
                \label{Eq::lemma17-2.2}
                \mathbf{e}'\cdot \mathbf{v_c}=u_2b_0 - u_s^{\prime} \cdot b_{s-2} - u_{s+1}^{\prime} \cdot b_{s-1} - u_{s+2}^{\prime} \cdot b_{s}=0.
            \end{equation}
            Since $s$ is odd and $3\leq s \leq \frac{n}{2}-4$, we can get $b_0=1$, $b_{s-2} = b_{s} = 0$, and $b_{s-1} = \frac{\xi^{s+3}-\xi^{-(s-1)}}{\xi^4-1}$ by Lemma \ref{lemma14}. 
            From the second equation of (\ref{Eq::lemma17-2.1}) and equation (\ref{Eq::lemma17-2.2}), we obtain
            \begin{equation*}
                \begin{cases}
                    u_2= -\frac{\xi^{2s+2}-1}{\xi^{2}(\xi^{2s-2}-1)}, \\
                    u_{s+1}' = \frac{(1-\xi^4)\xi^{s-3}}{\xi^{2s-2}-1}.       
                \end{cases}
            \end{equation*}
            Putting the values into the system (\ref{Eq::lemma17-2.1}), we have the following system:
            \begin{equation*}
                \begin{cases}
                    u_1  + u_{s} + u_{s+2}& = \frac{(\xi^2-1)(\xi^{s+1}+1)(\xi^{s-1}+1)}{\xi^2(\xi^{2s-2}-1)},  \\
                    u_1 - \xi^{s-1} u_s' - \xi^{s+1} u_{s+2}' &= 0,\\
                    u_1\xi^2 + u_{s}'\xi^{2s} + u_{s+2}'\xi^{2s+4} &= \frac{(\xi^{2s+2}+\xi^{2s}+\xi^{s-1}-1)(1+\xi^{s-1})(\xi^2-1)}{\xi^{2s-2}-1}.
                \end{cases}
            \end{equation*}
            By solving the system, we can get 
            \begin{equation*}
                \left\{
                \begin{aligned}
                    u_1 &= \frac{(\xi^2-1)(\xi^{2s}+\xi^{s+1}+\xi^{s-1}-1)}{\xi^2(\xi^{2s-2}-1)},\\
                    u_s' &= -\frac{\xi^{s+1}+1}{\xi^{s+1}(\xi^{s-1}+1)},\\
                    u_{s+2}' &= \frac{\xi^{s+1}-1}{\xi^{s+1}(\xi^{s-1}-1)}.
                \end{aligned}\right.
            \end{equation*}
            Note that $u_s'\in\mathbb{F}_q^*$,
            \begin{equation*}
                    (u_{s}^{\prime})^q  
                    = -\frac{\xi^{-(s+1)}+1}{\xi^{-(s+1)}(\xi^{-(s-1)}+1)} 
                    = -\frac{(1+\xi^{s+1})\xi^{s-1}}{1+\xi^{s-1}}
                    =-\frac{\xi^{s+1}+1}{\xi^{s+1}(\xi^{s-1}+1)}=u_s^{\prime}.
            \end{equation*}
            This means that $\xi^{2s}=1$, which is a contradiction for $3 \leq s \leq \frac{n}{2}-4 $.
        \end{itemize}
        In summary, there does not exist such a codeword in $\mathcal{C}$.
    \end{proof}
    Combining Lemmas \ref{lemma15}, \ref{lemma16}, and \ref{lemma17}, we can construct a class of MDS symbol-pair codes with $d_P = 9$.
    \begin{theorem}
        \label{Thm:theorem3}
        Suppose $q$ is an odd prime power and $n = 2q+2$. Let $\mathcal{C}$ be the cyclic code in $\mathbb{F}_q\left[x\right]/\langle x^n-1\rangle$ with generator polynomial $g(x) = m_{\xi^{-1}}(x)m_{\xi^{0}}(x)m_{\xi^{1}}(x)m_{\xi^{2}}(x)$. Then $\mathcal{C}$ is an $(n=2q+2,d_P=9)_q$ MDS symbol-pair code.
    \end{theorem}
    \begin{proof}
        When $q=3$, then $\mathcal{C}$ is a $[8,1,8]$ repetition code. It is clear that $d_P(\mathcal{C})\geq 9$. On the other hand, the Singleton-type bound (\ref{singleton_type_bound}) implies that $d_P\leq 9$. Thus $d_P = 9$ and $\mathcal{C}$ is an MDS symbol-pair code.
    
        When $q > 3$, note that $\dim(\mathcal{C})=n-\deg(g(x))=2q-5$ and $d_H (C) = 6$ by Lemma \ref{constr3::hammingdistance}, thus $\mathcal{C}$ is not an MDS code. Hence we can obtain $d_P (C) \geq 8$ by Lemma \ref{lem::Chen2017}. By the Singleton-type bound (\ref{singleton_type_bound}), it is sufficient to prove that $d_P(\mathcal{C}) \neq 8$, equivalently, there exists no codeword with minimum pair weight $8$. By contradiction, suppose that there is a codeword $\mathbf{c}\in\mathcal{C}$ with pair weight 8, then its certain cyclic shift must be one of the forms 
        \begin{gather*}
            (\star,\star,\star,\star,\star,\star,\star,\mathbf{0}_{n-7}),\\
            (\star,\star,\star,\star,\star,\mathbf{0}_{r},\star,\mathbf{0}_{n-r-6}) (1\leq r\leq n-7),\\
            (\star,\star,\star,\star,\mathbf{0}_{r},\star,\star,\mathbf{0}_{n-r-6}) (1\leq r\leq n-7),\\
            (\star,\star,\star,\mathbf{0}_{r},\star,\star,\star,\mathbf{0}_{n-r-6}) (1\leq r \leq n-7).
        \end{gather*}
        If $\mathbf{c}$ is of the first form, then the corresponding polynomial $c(x)$ has degree 6, while the generator polynomial $g(x)$ has degree 7, which leads to a contradiction. By Lemmas \ref{lemma15}, \ref{lemma16}, and \ref{lemma17}, no codeword in $\mathcal{C}$ can have the last three forms. Hence, we conclude that $\mathcal{C}$ is an $(n=2q+2,d_P=9)_q$ MDS symbol-pair code. 
    \end{proof}
\begin{remark}
To the best of our knowledge, the construction of MDS symbol-pair codes with $d_P = 9$ is given in \cite{Luo2023}, where the code length is $n=2m$ with $5 \leq m \leq q$. Thus, our construction achieves the longest code length for MDS symbol-pair codes with $d_P = 9$.
\end{remark}

\section{Conclusion}\label{sec4}
    Unlike classical MDS codes, constructing MDS symbol-pair codes with new parameters remains an important topic in coding theory. There are surely more families of interesting symbol-pair codes to be discovered. When the pair distance and code length are given, the dimension of an cyclic MDS symbol-pair code is determined, which in turn determines the number of roots of its generator polynomial. Therefore, the key to constructing MDS symbol-pair codes from single-root cyclic codes lies in selecting appropriate roots for the generator polynomial such that the minimum pair distance of the corresponding code meets the MDS bound. However, computing this minimum pair distance presents another challenge. In this paper, we utilized the decomposition of cyclic codes and the dual codes of each component code and analyzed the solutions of certain equations over finite fields to address these problems. Finally, we constructed three new classes of MDS symbol-pair codes from simple-root cyclic codes:
    \begin{itemize}
        \item [-] $(n=4q+4,k=4q-1,d_P = 7)_q$ cyclic code with $q\equiv 1\pmod 4$;
        \item [-] $(n=4q-4, k=4q-10, d_P=8)_q$ cyclic code with $q\equiv 3\pmod 4$;
        \item [-] $(n=2q+2,k=2q-5, d_P=9)_q$ cyclic code with $q$ is an odd prime power.
    \end{itemize}
Compared to previously known results, for 
$d_P=7$ or $8$, our $q$-ary MDS symbol-pair codes achieve the longest known code length when 
$q$ is not prime. For 
$d_P=9$, they remain the longest regardless of whether $q$ is prime or not. When using the simple-root cyclic code approach to construct MDS symbol-pair codes, the number of cases to analyze increases significantly as the pair distance or code length grows. This makes it particularly challenging to construct MDS symbol-pair codes with larger pair distances or longer code lengths.
\bigskip

\section*{Funding}
The research was supported in part by the National Key Research and Development Program of China under Grant Nos. 2022YFA1004900
and 2021YFA1001000, the National Natural Science Foundation of China under Grant No. 62201322, and the Natural Science Foundation of Shandong
Province under Grant No. ZR2022QA031, and the Taishan Scholar Program of Shandong Province





\bibliographystyle{splnc04}
\bibliography{ref}


\end{document}